\documentclass[
 aps,
 prd,
 twocolumn,
 superscriptaddress,
 nofootinbib,
 floatfix,
 showkeys,
 preprintnumbers
]{revtex4-2}

\usepackage[T1]{fontenc}
\usepackage{lmodern}
\usepackage{amsmath,amssymb,mathtools,bm}
\usepackage{booktabs}
\usepackage{microtype}
\usepackage{xcolor}
\usepackage{needspace}
\usepackage{etoolbox}
\usepackage{tikz}
\usetikzlibrary{decorations.pathmorphing,decorations.markings}
\usepackage{hyperref}
\usepackage{cleveref}
\usepackage{bbm}
\usepackage{comment}

\definecolor{BlueViolet}{rgb}{0.2, 0.00, 0.7}
\definecolor{Blue}{rgb}{0.15, 0.00, 0.9}
\definecolor{light-blue}{rgb}{0.3, 0.3, 1}
\definecolor{kit-green}{rgb}{0, 
0.58823 %150/255
, 0.50980 %130/255
}
\hypersetup{
  colorlinks=true,
  linkcolor=light-blue,
  citecolor=light-blue,
  urlcolor=kit-green,
  pdftitle={On-Shell Amplitudes and Black-Hole Perturbations: Exact Reissner--Nordstrom Mixing}
}

\newcommand{\dd}{\mathrm d}
\newcommand{\ii}{i}

\newcommand{\TF}{\mathrm{TF}}
\newcommand{\Order}{\mathcal O}
\newcommand{\Span}{\operatorname{span}}
\newcommand{\Tr}{\operatorname{tr}}
\newcommand{\comm}[2]{\left[#1,#2\right]}
\newcommand{\M}{\mathcal M}
\newcommand{\A}{\mathbf A}
\newcommand{\K}{\mathbf K}

\newcommand{\Pj}{\mathbf \Pi}
\newcommand{\slashed}[1]{\not{\!#1}}
\newcommand{\e}{\varepsilon}
\newcommand{\sgn}{\operatorname{sgn}}

\newcommand{\Jmat}{\begin{pmatrix}0&1\\-1&0\end{pmatrix}}
\newcommand{\cth}{c_\theta}
\newcommand{\sth}{s_\theta}
\DeclareMathOperator{\diag}{diag}

\newcommand{\BreakBibliographyBefore}[1]{%
  \let\RNoriginalbibitem\bibitem
  \RenewDocumentCommand{\bibitem}{o m}{%
    \ifstrequal{##2}{#1}{\newpage}{}%
    \IfNoValueTF{##1}
      {\RNoriginalbibitem{##2}}
      {\RNoriginalbibitem[##1]{##2}}%
  }%
}

\makeatletter
\@ifundefined{mathbbm}{}{}
\makeatother

\newtheorem{theorem}{Theorem}
\newtheorem{corollary}{Corollary}
\newtheorem{proposition}{Proposition}

\begin{document}
\preprint{CHIBA-EP-281}

\title{On-Shell Amplitudes and Black-Hole Perturbations: Exact Reissner--Nordstr{\"o}m Mixing}

\author{Kento Takahara}
\email{takahara.phys@gmail.com}
\affiliation{Center for Frontier Science, Faculty of Science, Chiba University}
\affiliation{Department of Physics, Graduate School of Science,
Chiba University, Chiba 263-8522, Japan}

\author{Teppei Kitahara}
\email{kitahara@chiba-u.jp}
\affiliation{Department of Physics, Graduate School of Science,
Chiba University, Chiba 263-8522, Japan}
\affiliation{
  Kobayashi-Maskawa Institute for the Origin of Particles and the
  Universe, Nagoya University,
  Furo-cho Chikusa-ku, Nagoya 464-8602 Japan
}

%\date{\today}

\begin{abstract}
Can flat-space on-shell amplitudes determine the channel basis of a coupled black-hole perturbation problem?
We address this question for electromagnetic and gravitational perturbations of a Reissner--Nordstr\"{o}m (RN) black-hole.
We organize the minimally coupled photon-graviton tree amplitudes off a heavy charged source into a $2\times 2$ channel-space matrix and perform a parity-resolved Jacob--Wick partial-wave projection. 
For every radiative multipole $\ell \geq 2$ and in both parity sectors, we show that the trace-free fixed-source partial-wave matrix is exactly proportional to the trace-free Moncrief coupling matrix, 
and therefore selects the same constant spectral projectors.
Through a first-Born matching, 
the amplitudes determine the same eigenspaces in the leading $r^{-3}$ weak-field potential, but not the complete radial potentials.
Using the exact classical RN potentials as independent curved-background input, 
we show that these projectors persist throughout the full radial domain. 
We also explicitly retain finite-mass effects through $\mathcal{O}(\omega/m)$, 
finding a nonvanishing commutator with the Moncrief coupling matrix, which shows that the RN-projector alignment is spoiled by genuine two-body recoil effects.
As a first step toward rotation,
we further extract the representation-independent linear-spin term from a minimally coupled Dirac amplitude.
We find that the complete tree-level channel matrix factorizes with a single linear-spin dressing, while the formal $J=2$ block fails to preserve the unchanged RN projectors. 
This restricted result does not constitute a test of Kerr--Newman separability, but it indicates that a rotating generalization must account for spin-induced angular-mode mixing.
We expect that this on-shell method can be extended to more general long-range scattering systems with two asymptotic channels.
\end{abstract}

\maketitle

\section{Introduction}
\label{sec:introduction}

The linearized Einstein--Maxwell system on a charged black-hole background is an example of a coupled wave system. 
In the spherically symmetric Reissner--Nordstr\"{o}m (RN) spacetime, however, gauge-invariant electromagnetic and gravitational perturbation variables can be transformed by {\it an \(r\)-independent linear rotation} into two master variables satisfying decoupled wave equations
\cite{Zerilli:1974ai,Moncrief:1974am,Moncrief:1974ng,Moncrief:1975sb, Chandrasekhar:1979iz,Chandrasekhar:1979rm}.
This transformation depends on the black-hole’s mass and the electric charge, as well as on the angular-momentum, and has traditionally been derived from the perturbation equations on the curved background.

Here we approach the problem from a different perspective and ask: to what extent is this decoupling structure already encoded in flat-space scattering amplitudes?
This question is useful for two reasons.
First, tree-level amplitudes isolate the leading long-range response without requiring a prior choice of gauge or master variables.
Second, they recast decoupling not as the problem of finding a suitable differential transformation, but as an algebraic condition on a matrix 
acting in the two-dimensional channel space 
spanned by the photon $(\gamma)$ and graviton $(g)$ states.

Extending the analysis to include rotation presents a potentially fundamental challenge. 
In a generic Kerr--Newman (KN) spacetime, 
the gravitoelectromagnetic perturbation system remains coupled, and no generally available pair of fully separated scalar radial equations serves as a counterpart to the RN master system.
Controlled descriptions are nevertheless available through slow-rotation expansions, coupled-mode calculations, and gauge-invariant formulations \cite{Pani:2013hpa,Pani:2013ijaDetailed,Dias:2015wqa,Giorgi:2020xvf,He:2023leo}. 
Meanwhile, the on-shell methods have revealed the remarkably compact spin dependence of black-hole amplitudes \cite{ArkaniHamed:2019ymq,Moynihan:2019bor,Chung:2019duq,Bautista:2021wfy,Bautista:2022wjf,Cangemi:2023ysz}, 
and spin universality structures in black-hole amplitudes \cite{Akpinar:2026oni}, including graviton-photon conversion in a charged system through quadratic order in spin \cite{Zheng:2026knmix}.

In this work, we separate the amplitude-side and curved-background parts of the argument. 
The fixed-source flat-space on-shell amplitudes determine the trace-free partial-wave operator and, through first Born matching, the channel eigenspaces of the leading $r^{-3}$ potential tail. 
Whether these eigenspaces persist in the complete radial problem is a separate question, which we answer using the exact classical RN perturbation potentials as independent input.

Our first result provides an operator-level criterion and establishes an exact result for the RN system.
Let
 \(\Lambda \equiv \ell (\ell +1 )- 2 = (\ell-1) (\ell +2 )\),
 and let \(P=\pm1\) denote polar $(P=+1)$ and axial parity $(P=-1)$. 
The RN perturbation equations, in a canonically normalized Moncrief channel basis, contain the  Moncrief coupling matrix 
\cite{Moncrief:1974am,Moncrief:1974ng,Moncrief:1975sb,Chandrasekhar:1979iz,Chandrasekhar:1979rm}:
\begin{align}
\begin{aligned}
 \K_\ell^P &=
 \begin{pmatrix}
 0&2P Q\sqrt{\Lambda}\\
 2P Q\sqrt{\Lambda}&6M
 \end{pmatrix}\,,
 \\
 \bigl(\K_\ell^P\bigr)^{\mathrm{TF}} &\equiv \K_\ell^P-3M \mathbf{1}_2 \,.
 \label{eq:intro-KC}
\end{aligned} 
\end{align}
Here, $\ell$ denotes the spherical-harmonic multipole number of the RN perturbations.
The second line defines the trace-free part \cite{Moncrief:1974am}, which plays a central role in this work.

We show that the trace-free part of the zeroth-order coefficient in the recoil expansion, reconstructed from flat-space amplitudes, is proportional to \((\K_\ell^P)^{\mathrm{TF}}\) in Eq.~\eqref{eq:intro-KC}. 
Since the trace part is proportional to the identity,
the amplitude-derived partial-wave matrix therefore 
has the same eigenspaces as
the RN  Moncrief coupling matrix, 
and its mixing angle agrees exactly with the Moncrief angle (up to sign conventions and interchange of the eigenvectors). 
In particular, we show 
\begin{equation}
 \comm{\A_{\ell,0}^P}{\K_\ell^P}=0 \,.
 \label{eq:intro-commutator}
\end{equation}
This statement concerns the fixed-source partial wave and its leading Born image. 
The all-radius diagonalization of the complete RN classical potentials by the same projectors is established independently in Sec.~\ref{sec:radial} using the exact curved-background equations.

The division of labor between the amplitude and curved-background analyses is important. 
Through the Born matching, the amplitude fixes the channel direction of the \(r^{-3}\) interaction, whereas the persistence of this direction in the exact classical RN potentials is established independently from the curved-background algebra. 
Consequently, the flat-space tree-level partial wave, the weak-field \(r^{-3}\) coefficient, and the exact curved-background potential all select the same pair of channel projectors.

This result applies to \(\ell\geq2\), where both radiative species are present. 
The \(\ell=0,1\) sectors are constrained or effectively single-channel sectors; their exceptional behavior does not represent a failure of decoupling.

To investigate separability on a KN background, our second result extracts the universal spin-dipole coefficient from a minimally coupled Dirac source.
On the aligned kinematic slice considered here, the explicit \(J=2\) block presented in Appendix~\ref{app:aligned} does not commute with the RN coupling matrices.
The RN projectors therefore cannot be applied to the rotating system without modification. 
This restricted diagnostic does not, however, constitute a test of separability in a generic KN spacetime, for the reasons explained in Section~\ref{sec:spin}.

The remainder of this paper is organized as follows. 
In Sec.~\ref{sec:rn-tree}, to formulate an amplitude-side description of photon and graviton Compton scattering by a RN black-hole, 
we use a heavy charged scalar as a spinless source matched to the RN spacetime parameters. 
We construct the complete minimally coupled tree-level amplitude matrix in two-channel space, 
encompassing elastic photon and graviton scattering as well as photon-graviton conversion processes. 
We then organize its heavy-source expansion into the fixed-source contribution and the leading recoil correction, 
and perform a parity-resolved Jacob--Wick partial-wave projection. 
We show that the trace-free part of the fixed-source partial-wave matrix selects precisely the Moncrief projectors
and determine the leading recoil-induced departure from this alignment.
In Sec.~\ref{sec:radial}, through first Born matching, we relate the amplitude result to the trace-free \(r^{-3}\) coefficient of the weak-field radial potential. 
Taking the exact classical RN perturbation potentials as independent curved-background input, 
we then prove that the same constant projectors identified from the amplitudes diagonalize the complete RN radial system throughout the full radial domain.
In Sec.~\ref{sec:spin}, as a first step toward extending the RN analysis to a rotating system, we replace the heavy scalar source with a minimally coupled Dirac field and compute the corresponding tree-level amplitudes, thereby extracting the universal \(g=2\) coefficient linear in spin. 
We further establish the factorization of the complete tree-level amplitudes on an aligned kinematic slice and analyze the formal \(J=2\) block. 
We then clarify the limitations of this restricted diagnostic as a test of separability in a generic KN spacetime.
Section~\ref{sec:discussion} summarizes our results and discusses possible future extensions. Appendices~\ref{app:spinor-scalar}--\ref{app:aligned} present, respectively, 
the spinor-helicity construction of the scalar tree-level amplitude matrix, details of the partial-wave projection and the radial Born matching, the Dirac reduction and proof of full-tree factorization, and the aligned \(J=2\) data together with the limitations of their interpretation.

\section{Flat-space amplitudes and the RN projector}
\label{sec:rn-tree}

\subsection{Minimally coupled effective Lagrangian}

To connect the flat-space on-shell amplitude calculation with the RN perturbation problem,
we represent the black-hole source by the simplest massive  state carrying the same asymptotic monopole data, namely its mass and electric charge. 
Since the RN black-hole is non-rotating,
we use a heavy charged scalar, $\Phi$, 
as the minimal amplitude-side representative of its asymptotic mass and charge.

Accordingly, we consider Einstein--Maxwell theory minimally coupled to a massive complex scalar field \(\Phi\) of electric charge \(q_\Phi\), namely scalar QED (SQED) coupled to dynamical gravity\footnote{%
Throughout this paper, we use the mostly-minus metric convention.}:
\begin{align}
\begin{aligned}
 S_{\mathrm{EM}\Phi}&=\int\dd^4x\,\sqrt{-g}\,\mathcal L_{\mathrm{EM}\Phi}\,,\\
 \mathcal L_{\mathrm{EM}\Phi}&=\frac{2}{\kappa^2}R-\frac14F_{\mu\nu}F^{\mu\nu}
 \\%[-1mm]
 &\quad+g^{\mu\nu}(D_\mu\Phi)^*D_\nu\Phi-m_\Phi^2|\Phi|^2\,,
 \label{eq:action}
 \end{aligned}
\end{align}
with $\kappa^2 = 32\pi G = 32\pi/M_{\mathrm{Pl}}^2$ and 
 \(D_\mu=\partial_\mu+\ii q_\Phi A_\mu\).
 The total electric charge of the RN black-hole is identified as $q_\Phi$.
 We expand the exterior metric around flat spacetime as
 \(g_{\mu\nu}=\eta_{\mu\nu}+\kappa h_{\mu\nu}\), where $h_{\mu\nu}$ is the canonically normalized graviton field.

The \(2\to2\) processes are
\begin{align}
\begin{aligned}
& \Phi_q(p)+I (k,h_i)\to 
 \Phi_q(p')+J (k',h_f)\quad  \textrm{with~} I,J \in\{\gamma,g\}.
 \label{eq:physical-process}
 \end{aligned}
\end{align}
Thus, the same calculation describes elastic photon and graviton scattering,
as well as the two conversion processes \(\gamma\leftrightarrow g\).

After matching to the heavy-source limit, \(\Phi_q\) represents a spinless source with RN parameters \((M,Q)\).
More precisely, 
only the exterior source data are matched. 
The scalar energy and current reproduce the linearized Einstein--Maxwell fields with the corresponding ADM mass and charge. 
No microscopic identification of the scalar state with a black-hole state is required.
This construction is the spin-zero limit of the standard minimal-coupling reconstruction of the long-range KN fields
\cite{Moynihan:2019bor,Chung:2019duq}.

On the other hand, 
for the RN black-hole background,
the spacetime metric takes the form \cite{Reissner:1916cle,Weyl:1917rtf,Nordstrom:1918cle}
\begin{align}
 \begin{aligned}
ds^2 &= f(r) d t^2 - \frac{d r^2}{f(r)} - r^2 d \Omega^2\,,\\
f(r) & \equiv 1 - \frac{2 M}{r} + \frac{Q^2}{r^2}\,,
\label{eq:metric}
 \end{aligned}
\end{align}
where $f(r)$ is the RN metric function in the mostly-minus metric convention.
Here, $M$ and $Q$ characterize the length scales associated with the mass and electric charge of the RN black-hole, respectively.
At tree level, these RN black-hole parameters can be identified with the scalar mass and charge as
\begin{equation}
 M=Gm_\Phi\,,
 \quad Q=\sqrt{\frac{G}{4\pi}}\,q_\Phi \,.
 \label{eq:MQ}
\end{equation}

Recoil effects are instead organized as an expansion in \(\omega/m_\Phi\), while
the factorized amplitude below is kept exact in \(m_\Phi\) until the wave
normalization is imposed. 

\subsection{Tree-level on-shell amplitudes}

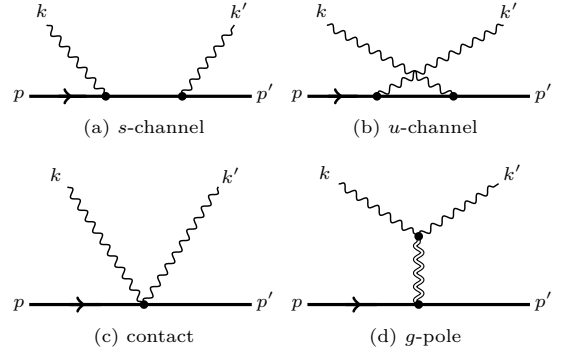
\begin{figure}[!t]
\centering
\begin{tikzpicture}[
 x=0.92cm,y=0.82cm,
scalar/.style={line width=1.2pt},
scalararrow/.style={
  line width=1.2pt,
  postaction={decorate},
  decoration={markings,mark=at position .5 with {\arrow{>}}}},
 photon/.style={decorate,decoration={snake,amplitude=.45mm,
  segment length=1.8mm},line width=.55pt},
 graviton/.style={
  decorate,
  decoration={
    snake,
    amplitude=.48mm,
    segment length=1.8mm
  },
  double=white,
  double distance=.75pt,
  line width=.42pt
},
 vertex/.style={circle,fill,inner sep=1.15pt},
 every node/.style={font=\scriptsize}
]
% s channel
\draw[scalar] (0,0)--(1.1,0)--(2.2,0)--(3.2,0);
\draw[scalararrow] (0.42,0)--(0.82,0);
\node at (-.15,.0) {\(p\)};
\node at (3.4,.05) {\(p'\)};
\draw[photon] (.25,1.15)--(1.1,0);
\draw[photon] (2.2,0)--(2.95,1.15);
\node[vertex] at (1.1,0) {};
\node[vertex] at (2.2,0) {};
\node at (.18,1.38) {\(k\)};
\node at (3.02,1.38) {\(k'\)};
\node at (1.65,-.52) {(a) \(s\)-channel};
% u channel
\draw[scalar] (4.0,0)--(5.0,0)--(6.1,0)--(7.2,0);
\draw[scalararrow] (4.30,0)--(4.72,0);
\node at (3.85,.0) {\(p\)};
\node at (7.4,.05) {\(p'\)};
\draw[photon] (4.28,1.15)--(6.1,0);
\draw[photon] (5.0,0)--(6.82,1.15);
\node[vertex] at (5.0,0) {};
\node[vertex] at (6.1,0) {};
\node at (4.20,1.38) {\(k\)};
\node at (6.90,1.38) {\(k'\)};
\node at (5.55,-.52) {(b) \(u\)-channel};
% second row
\begin{scope}[yshift=-0.5cm]
% scalar-QED contact
\draw[scalar] (0,-2.75)--(1.65,-2.75)--(3.2,-2.75);
\draw[scalararrow] (0.58,-2.75)--(1.03,-2.75);
\node at (-.15,-2.75) {\(p\)};
\node at (3.4,-2.70) {\(p'\)};
\draw[photon] (.55,-.85)--(1.65,-2.75);
\draw[photon] (1.65,-2.75)--(2.72,-.85);
\node[vertex] at (1.65,-2.75) {};
\node at (.38,-.66) {\(k\)};
\node at (2.90,-.66) {\(k'\)};
\node at (1.65,-3.28) {(c) contact};
% graviton pole
\draw[scalar] (4.0,-2.75)--(5.6,-2.75)--(7.2,-2.75);
\draw[scalararrow] (4.55,-2.75)--(5.00,-2.75);
\node at (3.85,-2.75) {\(p\)};
\node at (7.4,-2.70) {\(p'\)};
\draw[graviton] (5.6,-2.75)--(5.6,-1.65);
\draw[photon] (4.45,-.80)--(5.6,-1.65);
\draw[photon] (5.6,-1.65)--(6.75,-.80);
\node[vertex] at (5.6,-2.75) {};
\node[vertex] at (5.6,-1.65) {};
\node at (4.25,-.62) {\(k\)};
\node at (6.95,-.62) {\(k'\)};
\node at (5.55,-3.28) {(d) \(g\)-pole};
\end{scope}
\end{tikzpicture}
\caption{Tree-level contributions to
\(\Phi_q+\gamma\to\Phi_q+\gamma\).
The thick solid, wavy, and double-wavy lines denote the heavy charged scalar source matched to the RN parameters,
photon, and graviton, respectively.  The diagrams (a)--(c)
give the SQED contribution \(q_\Phi^2\mathcal C_0\), 
while the diagram (d)
gives the graviton-pole contribution
\(\M_{\gamma\gamma}^{(g{\rm -pole})}\).}
\label{fig:gamma-tree}
\end{figure}

We use physical \(2\to2\) kinematics of Eq.~\eqref{eq:physical-process} throughout this section.
We label the physical external legs as
$1(p)+2(k)\to3(p')+4(k')$,
with \(p+k=p'+k'\), \(p^2=p'^2=m_\Phi^2\), and \(k^2=k'^2=0\). The signs
\(h_i,h_f=\pm1\) label the two circular polarizations (right-handed circular polarization is the positive helicity).

For \(p+k=p'+k'\), define
\begin{align}
 D_s&=2p\cdot k,&
 D_u&=-2p\cdot k',\notag\\
 t&=(k-k')^2,&
 \Xi_\Phi&=m_\Phi^4-su=m_\Phi^2t-D_sD_u .
 \label{eq:invariants}
\end{align}
Momentum conservation also gives \(D_s+D_u+t=0\).
Assigning spinors directly to the physical null momenta \(k\) and \(k'\), 
\begin{equation}
 \begin{aligned}
 D_s&=\langle k|\boldsymbol{p}|k]\,,&
 D_u&=-\langle k'|\boldsymbol{p}|k']\,,\\
 t&%=-\langle kk'\rangle[k'k]
   =\langle k'k\rangle[k'k],&
 \Xi_\Phi&=\langle k'|\boldsymbol{p}|k]\,[k'|\bar{\boldsymbol{p}}|k\rangle \,.
 \end{aligned}
 \label{eq:invariants-spinor}
\end{equation}

Evaluating the SQED diagrams in Figs.~\ref{fig:gamma-tree}(a)--\ref{fig:gamma-tree}(c), we obtain the
gauge-invariant SQED Compton amplitude\footnote{%
Note that the superscripts on
amplitudes are ordered as \((h_f,h_i)\).}
\begin{align}
\mathcal{M}_{\gamma \gamma}^{\mathrm{SQED}} & \equiv q_{\Phi}^2 \mathcal{C}_0\,,\\
 \mathcal C_0^{h_f h_i}
 &=2\Biggl[
 \frac{(\e_i^{h_i}\cdot p)(\e_f^{\ast h_f}\cdot p')}{p\cdot k}
 -\frac{(\e_i^{h_i}\cdot p')(\e_f^{\ast h_f}\cdot p)}{p\cdot k'}\notag \\
 & 
 \quad -\e_i^{h_i}\cdot\e_f^{\ast h_f}\Biggr]\,.
 \label{eq:C0-covariant}
\end{align}
The first two terms contain the scalar-exchange poles, whereas the last is the seagull contact term. Projecting Eq.~(\ref{eq:C0-covariant}) onto the physical helicity states, we find
\begin{equation}
 \mathcal C_0^{++}=\frac{2\Xi_\Phi}{D_sD_u}\,,
 \quad
 \mathcal C_0^{+-}=-\frac{2m^2t}{D_sD_u}\,.
 \label{eq:C0-helicity}
\end{equation}
The remaining helicity components follow from parity. This result is
consistent with the scalar Compton amplitude of
Ref.~\cite{Ahmadiniaz:2019ppj}.
In spinor-helicity variables, these read
\begin{equation}
 \mathcal C_0^{++}
 =\frac{2\langle k'|\boldsymbol{p}|k]^2}{D_sD_u}\,,
 \quad
 \mathcal C_0^{+-}
 =\frac{2m^2\langle kk'\rangle^2}{D_sD_u}\,.
 \label{eq:C0-helicity-spinor}
\end{equation}
The remaining two components follow from parity.

Besides the pure SQED contributions, there is the graviton mediated $t$-channel diagram, see Fig.~\ref{fig:gamma-tree}(d).
We separately obtain
\begin{equation}
 \begin{aligned}
 \left(\M_{\gamma\gamma}^{(g{\rm -pole})}\right)^{++}
 &=-\frac{\kappa^2}{4t}
   \langle k'|\boldsymbol{p}|k]^2\,,
 &
 \left(\M_{\gamma\gamma}^{(g{\rm -pole})}\right)^{+-}
 &=0\,,
 \\[1mm]
 \left(\M_{\gamma\gamma}^{(g{\rm -pole})}\right)^{--}
 &=-\frac{\kappa^2}{4t}
   [k'|\bar{\boldsymbol{p}}|k\rangle^2\,,
 &
 \left(\M_{\gamma\gamma}^{(g{\rm -pole})}\right)^{-+}
 &=0\,.
 \end{aligned}
 \label{eq:gamma-g-pole}
\end{equation}
These expressions are consistent with the graviton-factorization results of
Refs.~\cite{Choi:1994ax,Holstein:2006bh,BjerrumBohr:2014lea,
Ahmadiniaz:2019ppj}.
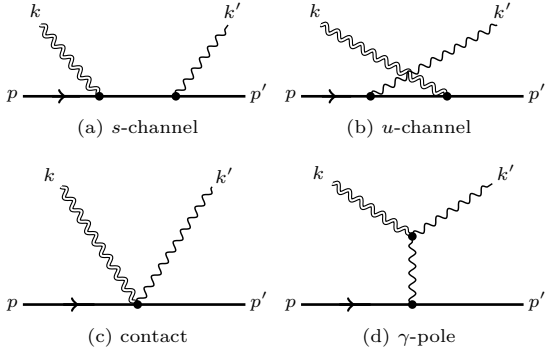
\begin{figure}[!t]
\centering
\begin{tikzpicture}[
 x=0.92cm,y=0.82cm,
 scalar/.style={line width=1.0pt},
 scalararrow/.style={
  line width=1.0pt,
  postaction={decorate},
  decoration={markings,mark=at position .5 with {\arrow{>}}}},
 photon/.style={decorate,decoration={snake,amplitude=.45mm,
  segment length=1.8mm},line width=.55pt},
 graviton/.style={
  decorate,
  decoration={
    snake,
    amplitude=.48mm,
    segment length=1.8mm
  },
  double=white,
  double distance=.75pt,
  line width=.42pt
 },
 vertex/.style={circle,fill,inner sep=1.15pt},
 every node/.style={font=\scriptsize}
]

% s channel
\draw[scalar] (0,0)--(1.1,0)--(2.2,0)--(3.2,0);
\draw[scalararrow] (0.42,0)--(0.82,0);
\node at (-.15,.0) {\(p\)};
\node at (3.4,.05) {\(p'\)};
\draw[graviton] (.25,1.15)--(1.1,0);
\draw[photon]   (2.2,0)--(2.95,1.15);
\node[vertex] at (1.1,0) {};
\node[vertex] at (2.2,0) {};
\node at (.18,1.38) {\(k\)};
\node at (3.02,1.38) {\(k'\)};
\node at (1.65,-.52) {(a) \(s\)-channel};

% u channel
\draw[scalar] (4.0,0)--(5.0,0)--(6.1,0)--(7.2,0);
\draw[scalararrow] (4.30,0)--(4.72,0);
\node at (3.85,.0) {\(p\)};
\node at (7.4,.05) {\(p'\)};
\draw[graviton] (4.28,1.15)--(6.1,0);
\draw[photon]   (5.0,0)--(6.82,1.15);
\node[vertex] at (5.0,0) {};
\node[vertex] at (6.1,0) {};
\node at (4.20,1.38) {\(k\)};
\node at (6.90,1.38) {\(k'\)};
\node at (5.55,-.52) {(b) \(u\)-channel};

% contact
\begin{scope}[yshift=-0.5cm] 
\draw[scalar] (0,-2.75)--(1.65,-2.75)--(3.2,-2.75);
\draw[scalararrow] (0.58,-2.75)--(1.03,-2.75);
\node at (-.15,-2.75) {\(p\)};
\node at (3.4,-2.70) {\(p'\)};
\draw[graviton] (.55,-.85)--(1.65,-2.75);
\draw[photon]   (1.65,-2.75)--(2.72,-.85);
\node[vertex] at (1.65,-2.75) {};
\node at (.38,-.66) {\(k\)};
\node at (2.90,-.66) {\(k'\)};
\node at (1.65,-3.28) {(c) contact};

% photon pole
\draw[scalar] (4.0,-2.75)--(5.6,-2.75)--(7.2,-2.75);
\draw[scalararrow] (4.55,-2.75)--(5.00,-2.75);
\node at (3.85,-2.75) {\(p\)};
\node at (7.4,-2.70) {\(p'\)};
\draw[photon]   (5.6,-2.75)--(5.6,-1.65);
\draw[graviton] (4.45,-.80)--(5.6,-1.65);
\draw[photon]   (5.6,-1.65)--(6.75,-.80);
\node[vertex] at (5.6,-2.75) {};
\node[vertex] at (5.6,-1.65) {};
\node at (4.25,-.62) {\(k\)};
\node at (6.95,-.62) {\(k'\)};
\node at (5.55,-3.28) {(d) \(\gamma\)-pole};
\end{scope}
\end{tikzpicture}
\caption{Tree-level contributions to
\(\Phi_q+g \to\Phi_q+\gamma\).
}
\label{fig:graviton-gamma}
\end{figure}

We now extend the Compton-scattering calculation above to include external gravitons. 
Besides graviton Compton scattering, 
we consider the two conversion processes:
\(\Phi_q+g \to\Phi_q+\gamma\) and \(\Phi_q+\gamma \to\Phi_q+ g \).
The tree diagrams for the former are shown in Fig.~\ref{fig:graviton-gamma}.
The diagrams for \(\Phi_q+\gamma \to\Phi_q+ g \) are obtained from them by interchanging the initial and final massless particles.
The rows below label the final species and the columns the initial species, in the order \((\gamma,g)\).

The resulting $2\times 2$ matrix amplitude is 
\begin{equation}
 \begin{aligned}
 &{\mathbf M}^{h_f h_i}
 \equiv
 \begin{pmatrix}
 \M_{\gamma\to\gamma}^{h_f h_i} & \M_{g\to \gamma}^{h_f h_i}\\
 \M_{\gamma\to g}^{h_f h_i}     & \M_{g\to g}^{h_f h_i}
 \end{pmatrix}
 \\[2mm]
 &=
  \begin{pmatrix}
 q_\Phi^2 \mathcal C_0^{h_fh_i}
 +\left(\M_{\gamma\gamma}^{(g{\rm -pole})}\right)^{h_fh_i}
 &
 \dfrac{\kappa q_\Phi}{2}N_i^{h_i}\mathcal C_0^{h_fh_i}
 \\[2mm]
 \dfrac{\kappa q_\Phi}{2}N_f^{h_f}\mathcal C_0^{h_fh_i}
 &
 \dfrac{\kappa^2}{8}R
 \left(\mathcal C_0^{h_fh_i}\right)^2
 \end{pmatrix}\,,
 \end{aligned}
 \label{eq:tree-matrix}
\end{equation}
with\footnote{%
Here, we define the stripped field-strength tensors by
\((F_i)^{\mu\nu}=k^\mu\e_i^\nu-k^\nu\e_i^\mu\) and
\((F_f^*)^{\mu\nu}
=k'^\mu\e_f^{*\nu}-k'^\nu\e_f^{*\mu}\).} 
\begin{align}
&N_i^{h_i}\equiv \frac{p\cdot F_i^{h_i}\cdot p'}{k\cdot k'}\,,\quad
N_f^{h_f} \equiv \frac{p'\cdot F_f^{\ast h_f}\cdot p}{k\cdot k'}\,,
\end{align}
and
\begin{align}
R\equiv\frac{D_sD_u}{2t}\,.
 \label{eq:N-R}
\end{align}
Note that although the graviton helicity is $\pm 2$, $h_{g}=\pm 1$ is taken in this expression, 
since the graviton polarization tensor can be written as \(\e_{h = \pm2} ^{\mu\nu}=\e_{h = \pm1}^\mu\e_{h = \pm1}^\nu\).
The corresponding spinor bracket expressions are
\begin{equation}
 \begin{aligned}
 N_i^+
 &=\frac{\langle k'|\boldsymbol{p}|k]}
 {\sqrt2\langle k'k\rangle}\,,
 &
 N_i^-
 &=\frac{[k'|\bar{\boldsymbol{p}}|k\rangle}
 {\sqrt2[k'k]}\,,
 \\[1mm]
 N_f^+
 &=-\frac{\langle k'|\boldsymbol{p}|k]}
 {\sqrt2[k'k]}\,,
 &
 N_f^-
 &=-\frac{[k'|\bar{\boldsymbol{p}}|k\rangle}
 {\sqrt2\langle k'k\rangle}\,,
 \\[1mm]
 R
 &=-\frac{
 \langle k|\boldsymbol{p}|k]\,
 \langle k'|\boldsymbol{p}|k']}
 {2\langle k'k\rangle[k'k]} \,.
 \end{aligned}
 \label{eq:N-R-spinor}
\end{equation}
This derivation is presented in
Appendix~\ref{app:spinor-scalar}. 
Equation~\eqref{eq:tree-matrix} is the familiar
factorization of gravitational Compton and photoproduction amplitudes
\cite{Choi:1994ax,Holstein:2006bh,BjerrumBohr:2014lea,
Ahmadiniaz:2019ppj}.
Note that the $\gamma$-pole ($g$-pole) contributions are already included in $\mathcal{M}_{g\to \gamma}$ and $\mathcal{M}_{\gamma \to g}$ ($\mathcal{M}_{g \to g}$) in the matrix \eqref{eq:tree-matrix} to ensure the correct  factorization.

\subsection{Recoil expansion}

We use the Jacob--Wick expansion \cite{Jacob:1959at} to decompose the
helicity amplitudes into partial waves.
Since the Jacob--Wick expansion is a two-body center-of-mass partial-wave expansion,
we work in the center-of-mass frame and parameterize the kinematics by the
massless-particle energy \(\omega\) and the scattering angle \(\theta\).

Let us define
\begin{equation}
 E=\sqrt{m_\Phi^2+\omega^2}\,,\quad
 E_{\rm cm}=E+\omega=\sqrt{s}\,,
\end{equation}
with
$p^\mu = (E, 0,0,-\omega)$ and $k^\mu = (\omega ,0,0,\omega)$, 
and the small recoil parameter 
\begin{equation}
 \epsilon_{\rm r}\equiv\frac{\omega}{m_\Phi }\,.
\end{equation}

Since all channels considered here contain the same massive particle and scattered  massless particles,
the magnitudes of the final-state three-momenta are equal to those of the initial-state three-momenta in the center-of-mass frame, even when $\mathcal{M}_{g\to \gamma}$ and $\mathcal{M}_{\gamma \to g}$.
This setup considerably simplifies the differential cross section, which becomes
\begin{align}
\begin{aligned}
    \frac{\dd\sigma}{\dd\Omega}
 & =\frac{1}{128\pi^2 E_{\rm cm} ^2} \sum_{h_i, h_f} 
 \left|\M_{I \to J}^{h_fh_i}\right|^2\\
 & =\frac{1}{2} \sum_{h_i, h_f} \left|f_{JI}^{h_f h_i}\right|^2\,,
 \end{aligned}
 \end{align}
with the rescaled helicity amplitudes 
\begin{align}
 f_{JI}^{h_f h_i} & \equiv \frac{\M_{I \to J}^{h_fh_i}}{8\pi E_{\rm cm}}\,.
\end{align}

We now substitute the matching conditions in Eq.~\eqref{eq:MQ}
\begin{equation}
 m_\Phi =\frac{M}{G}\,,\quad
 q_\Phi =\sqrt{\frac{4\pi}{G}}\,Q\,,\quad
 \kappa=\sqrt{32\pi G}\,,
\end{equation}
into the on-shell amplitudes, 
and expand in
\(\epsilon_{\rm r}\).
With \(\cth\equiv \cos(\theta/2)\),
\(\sth\equiv \sin(\theta/2)\), 
we obtain the following rescaled helicity amplitudes
\begin{align}
 \label{eq:rn-wave}
 f_{\gamma\gamma}^{++}
 &=M\frac{\cth^2}{\sth^2}
   \left(1+\epsilon_{\rm r}\right)-\frac{Q^2}{M}\cth^2
   \left[1-(1 - 2 s_\theta^2) \epsilon_{\rm r}\right] \notag\\
   &
   \quad +\Order(\epsilon_{\rm r}^2)\,,\\
 f_{\gamma\gamma}^{+-}&=-\frac{Q^2}{M}\sth^2
   \left[1-(3 - 2 s_\theta^2)\epsilon_{\rm r}\right]
   +\Order(\epsilon_{\rm r}^2)\,,\\
 f_{gg}^{++}
 &=M\frac{\cth^4}{\sth^2}
   \left[1+(1 + 2 s_\theta^2)\epsilon_{\rm r}\right]
   +\Order(\epsilon_{\rm r}^2)\,,\\
 f_{gg}^{+-}
 &=M\sth^2\left[1-(3 - 2 s_\theta^2)\epsilon_{\rm r}\right]
   +\Order(\epsilon_{\rm r}^2)\,,
   \label{eq:rn-wave-diagonal}\\
 f_{\gamma g}^{++}
 &=Q\frac{\cth^3}{\sth}
   \left(1+ 2 s_\theta^2\epsilon_{\rm r}\right)
   +\Order(\epsilon_{\rm r}^2)\,,\\
 f_{\gamma g}^{+-}
 &=-Q\cth\sth
   \left(1-2 c_\theta^2\epsilon_{\rm r}\right)
   +\Order(\epsilon_{\rm r}^2)\,.
 \label{eq:rn-wave-end}
\end{align}
With our choice of relative phase between the photon and graviton channels, 
time-reversal invariance implies
\begin{align}
    f_{g \gamma}^{h_f h_i}=f_{\gamma g}^{h_i h_f}\,.
\end{align}
Equations~\eqref{eq:rn-wave}--\eqref{eq:rn-wave-end}  retain the recoil expansion through first order in $\epsilon_{\rm r}$.
Their zeroth-order terms in $\epsilon_{\rm r}^0$ are the fixed-source amplitudes.
We derive the exact expressions and their partial-wave
projections in Appendix~\ref{app:projection}.

Note that the pure SQED contributions in Eq.~\eqref{eq:C0-covariant} 
 generate the pole-free Thomson terms
\begin{align}
 \label{eq:thomson}
 f_{\rm Th.}^{++}
 &=-\frac{Q^2}{M}\cth^2
   \left[1-(1 - 2 s_\theta^2) \epsilon_{\rm r}\right]
   +\Order(\epsilon_{\rm r}^2)\,,\\
 f_{\rm Th.}^{+-}
 &=-\frac{Q^2}{M}\sth^2
   \left[1-(3-2 s_\theta^2)\epsilon_{\rm r}\right]
   +\Order(\epsilon_{\rm r}^2)\,.
 \label{eq:thomson-end}
\end{align}
This term does not contribute to the partial-wave projection relevant to the fixed-source result below.\footnote{%
At this stage, we have not yet performed the partial-wave projection.
At zeroth order in recoil,
\(\cth^2=d^1_{11}\) and \(\sth^2=d^1_{1,-1}\), so the Thomson terms
contain only \(\ell=1\) Wigner-\(d\) functions. Their projections onto
all \(\ell\geq2\) partial waves therefore vanish by Wigner-\(d\)
orthogonality. Consequently, the fixed-source result below is unchanged
whether \(f_{\rm Th.}\) is included before the projection or treated
separately. At first order in recoil, however, the Thomson terms generate
an \(\ell=2\) odd-parity contribution, which we include explicitly in
the full recoil result below.}

As an independent check of our normalization,
using our geometrized RN variable conventions, we
 sum over final-state helicities for a fixed incident helicity.
The zeroth-order recoil terms in Eqs.~\eqref{eq:rn-wave}--\eqref{eq:rn-wave-end} (except for the Thomson terms) give
\begin{equation}
 \left.\frac{\dd\sigma_{\gamma\to\gamma}}{\dd\Omega}\right|_{g{\rm -pole}}
 =M^2\cot^4\frac{\theta}{2}\,,
 \quad
 \frac{\dd\sigma_{g\to g}}{\dd\Omega}
 =\frac{M^2}{\sth^4}\left(\cth^8+\sth^8\right)\,,
 \label{eq:cross-section}
\end{equation}
and
\begin{equation}
\frac{\dd\sigma_{\gamma\leftrightarrow g}}{\dd\Omega} \equiv \frac{\dd\sigma_{\gamma \to g}}{\dd\Omega} = 
 \frac{\dd\sigma_{g \to \gamma}}{\dd\Omega}
 =Q^2\cot^2\frac{\theta}{2}\left(\cth^4+\sth^4\right)\,.
 \label{eq:mixing-cross-section}
\end{equation}
Our first two results in Eq.~\eqref{eq:cross-section} 
are consistent with the standard low-frequency
electromagnetic and gravitational Schwarzschild cross sections
\cite{Matzner:1977lowfrequency,Crispino:2009em,Bautista:2021wfy}.
Our conversion cross section in Eq.~(\ref{eq:mixing-cross-section}) agrees with the established
Coulomb-field and RN results
\cite{Matzner:1976conversion,DeLogi:1977conversion,
OuldElHadj:2022conversion}.

Using
\(\cth^4+\sth^4=(1+\cos^2\theta)/2\) and
\(Q^2=Gq_\Phi^2/(4\pi)=\alpha_\Phi G\), the conversion result can equivalently be
written as
\begin{equation}
 \frac{\dd\sigma_{\gamma\leftrightarrow g}}{\dd\Omega}
 =\frac{\alpha_\Phi  G}{2}(1+\cos^2\theta)
  \cot^2\frac{\theta}{2}\,.
 \label{eq:conversion-RS-check}
\end{equation}
This form agrees with the zero-recoil limit of the scalar-target result
of Ravndal and Sundberg \cite{Ravndal:2001nb}. Their inverse process has the same cross section, as required 
by the relation \(f_{\gamma g}=f_{g\gamma}\).

From Eqs.~\eqref{eq:thomson}--\eqref{eq:thomson-end}, we independently calculate the
helicity-reversing cross section and obtain
\begin{align}
\begin{aligned}
\left.\frac{\dd\sigma_{\gamma\to\gamma}^{\rm flip}}{\dd\Omega}
 \right|_{\rm Th.}
 &\equiv
\frac12\sum_{h_i,h_f=\pm1}
\delta_{h_f,-h_i}
\left|f_{\mathrm{Th.}}^{h_fh_i}\right|^2\\
 &=\frac{Q^4}{M^2}
 s^4_\theta\,.
 \end{aligned}
\end{align}
This result is consistent with the charge-induced low-frequency RN
backscattering term of Ref.~\cite{Crispino:2014rnbackscatter}.
That reference further identifies the contribution as arising solely
from \(\ell=1\), which explains why it does not enter the \(\ell\geq2\) projector theorem.

\subsection{Regulated Jacob--Wick projection}

Following Jacob and Wick \cite{Jacob:1959at}, we normalize the helicity
partial waves. 
%Let \(h_i,h_f=\pm1\) denote the helicity signs.
For an
initial species \(I\) and a final species \(J\) with \(I,J\in\{\gamma,g\}\), we define the physical
helicities entering the Wigner functions by
\begin{equation}
    \mu_i=s_I h_i\,,\quad
 \mu_f=s_J h_f\,,
\end{equation}
with
\begin{equation}
 s_\gamma\equiv 1\,,\quad s_g\equiv 2\,.
\end{equation}
Here, we write $f_{\mu_f\mu_i}^{JI} \equiv f_{JI}^{h_f h_i}.$
We normalize the helicity partial wave as
\begin{align}
 f_{\mu_f \mu_i}^{JI}(\theta,\phi)&=
 \frac{e^{i (\mu_i-\mu_f)\phi}}{2\ii\omega}
 \sum_\ell(2\ell+1)\notag \\
 &\quad \times
 \left(S_{\ell;\mu_f \mu_i}^{JI}-\delta^{JI}\delta_{\mu_f \mu_i}\right)
 d^\ell_{\mu_f \mu_i}(\theta)\,,
 \label{eq:JW}
\end{align}
with the inverse relation
\begin{equation}
 S_{\ell;\mu_f \mu_i}^{JI}-\delta^{JI}\delta_{\mu_f \mu_i}
 =i \omega\int_{-1}^{1}\dd z\,
 d^\ell_{\mu_f \mu_i}(\theta)f_{\mu_f \mu_i}^{JI}(z)\,,
 \label{eq:JW-inverse}
\end{equation}
where $\omega$ denotes the energy of the incident wave and $z \equiv \cos \theta $.
Here and below, the inverse relation is evaluated at \(\phi=0\). 
Note that the convention for the Wigner small-$d$ matrix function \(d^\ell_{m'm} (\theta)\) also assigns the first index to the final 
helicity and the second to the initial helicity.
We define the partial-wave $S$-matrix element by
\begin{equation}
 \mathbf S_\ell^P-\mathbf{1}_2
 \equiv i\omega\,\A_\ell^P \,.
\end{equation}
The component form of the parity projection is given in Appendix~\ref{app:projection}.

Using the intrinsic parity phases \(\eta_\gamma=-1\) and \(\eta_g=+1\), we define
the parity eigenstates
\begin{equation}
 |I;P,\ell m\rangle=
 \frac{1}{\sqrt2}
 \Bigl(|+\!s_I;\ell m\rangle+\eta_IP|-\!s_I;\ell m\rangle\Bigr)\,.
 \label{eq:parity-state}
\end{equation}
To match the standard RN master-variable convention, we rephase the
electromagnetic parity state using \(\mathbf U_P=\diag(P,1)\). This assigns the sign \(P\) to both conversion entries and brings the channel-coupling matrix into the canonical form in Eq.~\eqref{eq:intro-KC}, without changing any probability or eigenspace.

The diagonal amplitudes in
Eqs.~\eqref{eq:rn-wave}--\eqref{eq:rn-wave-end} share the same forward
logarithmic divergence. Since the channel direction is extracted from the
full coupled matrix, the two channels must be regulated consistently. We
therefore impose a single cutoff \(z\leq1-\delta\)  on the entire matrix.\footnote{%
$\delta \to 0$ corresponds to $\theta \to 0$ in the Compton scatterings.}
The divergent contribution is proportional to the identity matrix, 
and hence drops out of both the trace-free part and every channel commutator, as will be seen below.
The remaining
finite Wigner-\(d\) integrals are listed in
Appendix~\ref{app:projection}.

To distinguish the fixed-source result from finite-mass recoil effects, we
expand the partial wave as
\begin{equation}
 \A_{\ell,\rm long}^P(\epsilon_{\rm r})
 =\A_{\ell,0}^P
 +\epsilon_{\rm r}\A_{\ell,\rm r}^P
 +\Order(\epsilon_{\rm r}^2)\,.
 \label{eq:recoil-partial-wave-expansion}
\end{equation}
The first term \(\A_{\ell,0}^P\) is the fixed-source contribution relevant
to the RN background problem, whereas
\(\A_{\ell,\rm r}^P\) contains the leading finite-mass two-body correction.

\Needspace{10\baselineskip}
\begin{theorem}[Fixed-source channel projectors]
\label{thm:RN-tree}
For every radiative multipole \(\ell\geq2\), the trace-free
fixed-source partial-wave matrix obtained from
Eqs.~\eqref{eq:rn-wave}--\eqref{eq:rn-wave-end} satisfies
\begin{align}
 \label{eq:RN-tree-result}
 \bigl(\A_{\ell,0}^{+}\bigr)^{\TF}
 &=
 \frac{\Lambda+4}{(\Lambda+2)\Lambda}
 \bigl(\K_\ell^+\bigr)^{\TF}\,,
 \\
 \bigl(\A_{\ell,0}^{-}\bigr)^{\TF}
 &=
 \frac{1}{\Lambda+2}
 \bigl(\K_\ell^-\bigr)^{\TF}\,,
 \label{eq:RN-tree-result-end}
\end{align}
with \(\Lambda \equiv \ell (\ell +1 )- 2 = (\ell-1) (\ell +2 )\).
The matrix $(\K_\ell^P)^{\TF}$ is defined in Eq.~\eqref{eq:intro-KC}.
These relations imply algebraically that
\begin{equation}
 \comm{\A_{\ell,0}^P}{\K_\ell^P}=0 \,.
 \label{eq:RN-comm}
\end{equation}
Thus \(\A_{\ell,0}^P\) and \(\K_\ell^P\) have the same eigenspaces in
channel space.
\end{theorem}

\begin{corollary}[Moncrief rotation angle]
\label{cor:Moncrief-angle}
\mbox{}\par\nobreak\smallskip\noindent
Adopt the canonical \((\gamma,g)\) phase convention, and denote the matrix
elements of \(\A_{\ell,0}^P\) by \(A_{JI,0}^P\). Let
\(\chi_{\ell,0}^P\) and \(\chi_{\ell,\mathrm{Mon.}}^P\) be the angles that
diagonalize \(\A_{\ell,0}^P\) and \(\K_\ell^P\), respectively, with their
signs fixed by the first equality below. Then
\begin{equation}
 \begin{aligned}
 \tan\!\left(2\chi_{\ell,0}^P\right)
 &=\frac{2A_{\gamma g,0}^P}
 {A_{\gamma\gamma,0}^P-A_{gg,0}^P}\\
 &=-\frac{2P Q\sqrt{\Lambda}}{3M}
 =\tan\!\left(2\chi_{\ell,\mathrm{Mon.}}^P\right)\,.
 \end{aligned}
 \label{eq:ratio}
\end{equation}
On the curved-background side, the Moncrief variables are conventionally
related by an \(r\)-independent rotation through the angle
\(\chi_{\ell,\mathrm{Mon.}}^P\)
\cite{Moncrief:1974am,Moncrief:1974ng,Moncrief:1975sb,
Crispino:2014rnbackscatter,OuldElHadj:2022conversion}. In the present phase
convention, these references give
\begin{equation}
 \begin{aligned}
 \sin\!\left(2\chi_{\ell,\mathrm{Mon.}}^P\right)
 &=-\frac{2P Q\sqrt{\Lambda}}{\sqrt{9M^2+4Q^2\Lambda}}\,,\\
 \cos\!\left(2\chi_{\ell,\mathrm{Mon.}}^P\right)
 &=\frac{3M}{\sqrt{9M^2+4Q^2\Lambda}}\,.\\
 %d_\ell&=\sqrt{9M^2+4Q^2\Lambda}.
 \end{aligned}
\end{equation}
Taking the ratio of the first two expressions, we find the same
angular relation as in Eq.~\eqref{eq:ratio}. Comparing the two results,
we conclude that the mixing angle obtained from the  flat-space on-shell 
amplitudes agrees exactly with the Moncrief angle. The on-shell 
calculation therefore identifies the same rotation in channel space
as the standard Moncrief analysis.
\end{corollary}

The theorem and its corollary establish the channel rotation selected by the fixed-source coefficient \(\A_{\ell,0}^P\). We now test whether the
same rotation persists at first order in the finite-mass recoil expansion.
To this end, we evaluate the commutator of the recoil-corrected partial-wave matrix with \(\K_\ell^P\), first for the long-range pole contribution and
then with the pole-free Thomson terms included.

\begin{proposition}[First-order recoil commutator]
\label{prop:RN-recoil}
Separating the long-range pole and Thomson contributions in
Eqs.~\eqref{eq:rn-wave}--\eqref{eq:rn-wave-end},
we define
\begin{equation}
 \A_{\ell,\mathrm{full}}^P
 \equiv
 \A_{\ell,\mathrm{long}}^P+\A_{\ell,\mathrm{Th.}}^P \,.
\end{equation}
Its commutator is
\begin{align}
 \comm{\A_{\ell,\rm long}^P(\epsilon_{\rm r})}{\K_\ell^P}
 &=\epsilon_{\rm r}\,M Q\,\rho_{\ell P}^{\rm long}\,\Jmat
 \notag \\
& \quad  +\Order(\epsilon_{\rm r}^2)\,,
 \label{eq:RN-recoil-commutator}
\end{align}
where the dimensionless $\ell$-dependent coefficient $\rho_{\ell,\rm{P}}^{\rm long}$ is 
\begin{equation}
 \rho_{\ell,+}^{\rm long}
 =-\frac{12\sqrt{\Lambda}}{\Lambda+2}
  -\frac85\delta_{\ell2}\,,
 \quad
 \rho_{\ell,-}^{\rm long}
 =\frac{12(\Lambda+4)}{(\Lambda+2)\sqrt{\Lambda}} \,.
 \label{eq:RN-recoil-rho}
\end{equation}
Here \(\delta_{\ell2}\) denotes the Kronecker delta. In particular,
\[
 \rho_{2,+}^{\rm long}=-\frac{28}{5}\,,
 \quad
 \rho_{2,-}^{\rm long}=8\,,
\]
while at large \(\ell\),
\[
 \rho_{\ell,+}^{\rm long}\sim-\frac{12}{\ell}\,,
 \quad
 \rho_{\ell,-}^{\rm long}\sim\frac{12}{\ell}\,.
\]
If the pole-free Thomson contributions
\eqref{eq:thomson}--\eqref{eq:thomson-end} are included, only the
\((\ell,P)=(2,-)\) result is modified:
\begin{equation}
 \rho_{2,-}^{\rm full}
 =8\left(1-\frac{Q^2}{5M^2}\right)\,,
 \quad
 \rho_{\ell P}^{\rm full}=\rho_{\ell P}^{\rm long}
 \quad\text{otherwise}.
 \label{eq:RN-recoil-full}
\end{equation}
\end{proposition}
\vspace{1\baselineskip}

The nonzero commutator in Eq.~\eqref{eq:RN-recoil-commutator}
reflects a finite-mass two-body recoil effect.
Appendix~\ref{app:projection} gives the angular integrals leading to
Eqs.~\eqref{eq:RN-recoil-rho} and \eqref{eq:RN-recoil-full}.

Finally, we summarize what follows directly from the amplitude analysis. For every radiative multipole \(\ell \geq 2\), we show that the trace-free part of the fixed-source partial-wave matrix is proportional to \((\K_\ell^P)^{\mathrm{TF}}\). It therefore commutes with \(\K_\ell^P\) and has the same eigenspaces. Equivalently,
the channel rotation extracted from the flat-space amplitudes agrees with the Moncrief rotation, up to phase conventions and interchange of the eigenvectors. The first recoil correction generally has a nonzero commutator, showing that this agreement applies specifically to the fixed-source coefficient \(\A_{\ell,0}^P\). These conclusions concern the
projected amplitudes. The all-radius diagonalization of the exact RN perturbation equations is established independently in Sec.~\ref{sec:radial}.

\section{Born matching and the exact RN potential algebra}
\label{sec:radial}
\subsection{First-Born information extracted from the amplitude}
\label{subsec:First-Born}

Through the first Born relation, it is found that the amplitude-derived partial-wave matrix elements determine the trace-free and channel-dependent coefficient, ${\mathbf W}$, of the leading \(r^{-3}\) term in the weak-field radial potential:
\begin{equation}
\left[
 \left( -\frac{\dd^2}{\dd r_*^2}
 +\frac{\Lambda+2}{r^2}\right)\mathbf 1_2
 +\frac{\mathbf W}{r^3}
\right]\boldsymbol{\psi}
=\omega^2\boldsymbol{\psi} \,.
 \label{eq:radial-sign}
\end{equation}
Here, $\boldsymbol{\psi}=\left(\psi_\gamma, \psi_g\right)^T$ 
is the two-component radial field, and 
$r_*$ is the RN tortoise coordinate defined by 
$d r_* / d r=f(r)^{-1}$.
No higher-order radial information is inferred from the tree-level amplitude.
The Riccati--Bessel integral derived in Appendix~\ref{app:projection} gives  the first Born relation
\begin{equation}
 {\mathbf S}_\ell^P
 =\mathbf 1_2+\ii\omega\mathbf A_{\ell,0}^P+\cdots\,,
 \quad
 \left(\mathbf A_{\ell,0}^P\right)^{\TF}
 =-\frac{\left(\mathbf W_\ell^P\right)^{\TF}}{\Lambda+2}\,.
 \label{eq:born}
\end{equation}
The ellipsis denotes second- and higher-Born iterations of the radial potential, 
rather than recoil corrections.
Inverting the first-Born relation in Eq.~\eqref{eq:born} using the
amplitude results in Eqs.~\eqref{eq:RN-tree-result}--\eqref{eq:RN-tree-result-end},
we obtain
\begin{align}
\begin{aligned}
 \left(\mathbf W_{\ell,\rm e}^{\mathrm{amp}}\right)^{\TF}
 &=- \frac{\Lambda + 4}{\Lambda} 
  \bigl(\K_\ell^+\bigr)^{\TF},\\
  \left(\mathbf W_{\ell,\rm o}^{\mathrm{amp}}\right)^{\TF}
 &=-\bigl(\K_\ell^-\bigr)^{\TF}.
 \label{eq:W-weak}
\end{aligned} 
\end{align}
Here and following, the even (e) and odd (o) labels in the (exact) classical potential corresponds to 
$P=+1$ (polar) and $P=-1$ (axial parity) sectors in the previous section.
These expressions are the on-shell amplitude predictions for the
trace-free, channel-dependent part of the leading \(r^{-3}\) weak-field
interaction. 

Independently, the large-\(r\) expansion of the exact RN
potentials gives the same coefficients, including the relative
normalization between the even- and odd-parity sectors.
This Born matching is an expansion in the strength of the background
interaction, not a low-frequency expansion, and therefore does not
require \(\omega M\ll1\). At tree level,
it determines neither the
channel-independent part of the potential nor the higher powers of
\(1/r\). We now introduce the exact RN potentials as independent
curved-background input and determine whether the amplitude-derived
projectors persist throughout the full radial problem.

\subsection{Independent all-radius closure from the exact RN potentials}
\label{subsec:exact-RN-closure}

We now leave the on-shell amplitude analysis, and instead introduce the exact RN
perturbation potentials as independent curved-background input.
The following all-radius argument is therefore not a consequence of the
tree-level amplitude.
Its purpose is to determine whether the channel direction inferred from the leading Born tail persists in the complete radial problem.

For each fixed \((\ell,P)\), the trace-free part of the channel matrix
\(\K_\ell^P\) satisfies 
\begin{equation}
 \left[\bigl(\K_\ell^P\bigr)^{\mathrm{TF}}\right]^2=\left(9M^2+4Q^2\Lambda\right)\mathbf{1}_2\,,
 \label{eq:minpoly}
\end{equation}
This follows immediately from the Cayley--Hamilton theorem for a $2\times 2$ matrix.
It therefore defines the constant spectral projectors
\begin{align}
    \begin{aligned}
 \Pj_{\ell,\pm}^P
 &=
 \frac12\left[
 \mathbf{1}_2\pm\frac{1}{\sqrt{9M^2+4Q^2\Lambda}}\bigl(\K_\ell^P\bigr)^{\mathrm{TF}}
 \right]\,,\\[1mm]
 k_\pm & =3M\pm \sqrt{9M^2+4Q^2\Lambda}\,,
 \label{eq:projectors}
 \end{aligned}
\end{align}
with
\begin{align}
\mathbf{K}_{\ell}^P \Pj_{\ell, \pm}^P& =k_{ \pm} \Pj_{\ell, \pm}^P\,, \quad & \Pj_{\ell,+}^P+\Pj_{\ell,-}^P &=\mathbf{1}_2\,,\\
\left(\Pj_{\ell, \pm}^P \right)^2 & =\Pj_{\ell, \pm}^P\,, \quad & \Pj_{\ell, +}^P \Pj_{\ell, -}^P& =0\,,
\end{align}
where \(k_\pm\) are the two eigenvalues of \(\K_\ell^P\).
These projectors depend on \((M,Q,\ell,P)\), but not on \(r\).

The exact classical odd-parity matrix potential can be expressed as
\begin{equation}
 \mathbf V_{\ell,{\rm o}}(r)
 =
 f(r)\left[
 \left(
 \frac{\Lambda+2}{r^2}+\frac{4Q^2}{r^4}
 \right)\mathbf{1}_2
 -
 \frac{\K_\ell^{-}}{r^3}
 \right]\,,
 \label{eq:Vodd}
\end{equation}
where $f(r)$ is the RN metric function in Eq.~\eqref{eq:metric}.
In the corresponding canonically rephased channel basis, 
the even-parity potential can be written as
\begin{align}
 \mathbf V_{\ell,{\rm e}}(r)
 &=
\mathbf U_{-} \mathbf V_{\ell,{\rm o}}(r)\mathbf U_{-}
 \notag \\
 & \quad +
 2\K_\ell^{+}
 \frac{\dd}{\dd r_*}
 \left[
 \frac{f(r)}{r}
 \left(
 \Lambda r\mathbf{1}_2+\K_\ell^{+}
 \right)^{-1}
 \right]\,.
 \label{eq:Veven}
\end{align}
It should be noted here that ${\mathbf U_{-}}\K_\ell^{-}{\mathbf U_{-}}= \K_\ell^{+} $ in the first term.

For the general $2\times 2$ matrix $\K$, the Cayley--Hamilton theorem gives the identity:
\begin{align}
\left( \lambda \mathbf{1}_2 +\mathbf{K}\right)^{-1}=\frac{(\lambda+\operatorname{tr} \mathbf{K}) \mathbf{1}_2-\mathbf{K}}{\lambda(\lambda + \operatorname{tr} \mathbf{K}) + \operatorname{det} \mathbf{K}}\,,
\end{align}
with $\lambda \in \mathbb{C}$ such that  $\operatorname{det}(\lambda \mathbf{1}_2+\mathbf{K}) \neq 0$. 
For \(\K=\K_\ell^{+}\) and $\lambda = \Lambda r$,
this identity becomes 
\begin{equation}
 \left(\Lambda r\mathbf{1}_2+\K_\ell^{+} \right)^{-1}
 =
 \frac{
 (\Lambda r + 6M)\mathbf{1}_2-\K_\ell^{+}
 }{
 \Lambda r (\Lambda r +6M)-4Q^2\Lambda
 }\,.
 \label{eq:resolvent}
\end{equation}
These show that
the resolvent of the exact classical potential is a linear combination of \(\mathbf{1}_2\) and
\(\K_\ell^P\). Since \(\K_\ell^P\) is independent of \(r\), the
\(r_*\)-derivative in Eq.~\eqref{eq:Veven} differentiates only the
scalar coefficient functions in this decomposition. 

Consequently, for
each parity sector \(s={\rm o}, {\rm e}\),
\begin{equation}
 \mathbf V_{\ell,s}(r)
 =
 v_{\ell,s,0}(r)\mathbf{1}_2
 +
 v_{\ell,s,1}(r)\K_\ell^P
 \in
 \Span\left\{
 \mathbf{1}_2,\K_\ell^P
 \right\}.
 \label{eq:exact-algebra-decomposition}
\end{equation}

Using the constant spectral projectors, 
the same statement can be written as
\begin{equation}
 \mathbf V_{\ell,s}(r)
 =
 V_{\ell,s,+}(r)\, \Pj_{\ell,+}^P
 +
 V_{\ell,s,-}(r)\, \Pj_{\ell,-}^P\,.
 \label{eq:exact-spectral-decomposition}
\end{equation}
It follows immediately that
\begin{align}
 \comm{\mathbf V_{\ell,s}(r)}{\K_\ell^P}&=0\,,\\
 \comm{
 \mathbf V_{\ell,s}(r)
 }{
 \mathbf V_{\ell,s}(r')
 }& =0
 \quad
 \text{for all }r,r'.
 \label{eq:exact-algebra}
\end{align}
Because the radial kinetic operator is proportional to the identity in
channel space and the projectors \(\Pj_{\ell,\pm}^P\) are independent
of \(r\), no derivative mixing is generated by this transformation.
The exact classical RN radial system is therefore separated, at every radius,
into two independent master equations by the same constant projectors.

This is the curved-background part of the argument. Combining it with
the amplitude/Born result of the previous  subsection, we conclude that
the candidate projectors selected by the fixed-source flat-space tree-level 
amplitudes are precisely the projectors that diagonalize the complete
RN radial potentials.

The logical status of the two results should be distinguished. The
tree-level amplitude determines the trace-free, channel-dependent part
of the leading \(r^{-3}\) coefficient and reconstructs its spectral
projectors. It does not determine the channel-independent part of the
potential, the higher powers of \(1/r\), or the complete radial
functions \(V_{\ell,s,\pm}(r)\). Those additional potentials, 
together with
the proof of all-radius closure, 
are supplied here by the exact
classical RN perturbation equations.

\section{Promotion to arbitrary-spin and aligned Dirac kernel}
\label{sec:spin}

To extend the amplitude construction from the RN
background to the Kerr--Newman (KN) background, the
heavy source must carry not only mass and electric charge but also angular momentum.  
At linear order in spin, 
rotation generates gravitomagnetic and
magnetic-dipole couplings that cannot be encoded by the scalar source used
for the RN Compton scatterings.
A spinning source is therefore required on the amplitude side.

The simplest choice is a minimally coupled massive Dirac field, $\Psi$, 
for which minimal electromagnetic coupling fixes the spin-dipole interaction to the Dirac value \(g=2\). We use the Dirac field only as a coefficient extractor. 
It will be shown that 
the coefficient linear in the classical spin vector is independent of the massive-spin representation. 
It could therefore be extracted from a
spin-\(\tfrac12\) amplitude and promoted to arbitrary spin, although this
does not determine quadratic or higher-spin terms.
We perform this
extraction on the aligned kinematic slice. 
Let $S^\mu\equiv W^\mu/m_\Psi$ denote the covariant spin vector of the massive source, 
where $W_\mu \equiv \frac12 \varepsilon_{\mu\nu \rho\sigma}J^{\nu\rho}P^\sigma$ is the Pauli--Lubanski pseudovector. 
In the source rest frame, it takes the form $S^\mu=(0,\boldsymbol S)$. 
Introducing the mass-rescaled spin vector
\begin{align}
a^\mu\equiv \frac{S^\mu}{m_\Psi}\,,
\label{eq:amu}
\end{align}
we match the source parameters to the
KN quantities according to
\begin{equation}
 \begin{gathered}
 M=Gm_\Psi\,,\quad
 Q=\sqrt{\frac{G}{4\pi}}\,q_\Psi\,,\quad
 a\equiv \sqrt{-a_\mu a^\mu}  = \frac{\hbar j}{m_\Psi}\,.
 \end{gathered}
 \label{eq:correlated-classical-limit}
\end{equation}
Here, $j$ denotes the spin quantum number of the massive source and $j=\frac{1}{2}$ for the Dirac field.
Note that we denote the total angular momentum $J$ of the Jacob--Wick expansion in this section, 
while reserving $\ell$ for the spherical RN multipole. (In the non-rotating limit, they should be identified.)
In the amplitude formulae we set \(\hbar=1\); with \(\hbar\) restored, the two
independent dimensionless expansion parameters are
\begin{equation}
 \epsilon_{\rm r}\equiv\frac{\hbar\omega}{m_\Psi}\,,
 \quad
 \epsilon_{\rm spin}\equiv\omega a \,,
 \label{eq:two-small-parameters}
\end{equation}
and \(\epsilon_{\rm spin}=\omega a\) on the aligned slice.  
Since $\epsilon_{\mathrm{spin}}=j \epsilon_{\mathrm{r}}$, keeping $\epsilon_{\mathrm{spin}}$ finite as $\epsilon_{\mathrm{r}} \rightarrow 0$ requires the correlated large-spin limit $j, m_\Psi \to \infty$ with $a=\hbar j / m_\Psi$ fixed. 
Here, we use the Dirac field to compute the spin-universal term.
At fixed \(j\), by contrast, \(m_\Psi\to\infty\)
forces \(a\to0\); a single spin-\(\tfrac12\) state cannot itself realize the
finite-\(a\) limit.
In the rest of this section,
\(\Order (\epsilon_{\rm spin}^n)\) denotes terms of at least
\(n\)-th order in spin at fixed recoil, whereas
\(\Order (\epsilon_{\rm r}^n)\) denotes terms with at least \(n\)
powers of recoil after the spin-even/odd projection.  Mixed terms such as
\(\epsilon_{\rm spin}\epsilon_{\rm r}\) belong to the latter class.  This
notation is always relative to the displayed leading amplitude.
For compactness, set
\begin{equation}
 \delta_{\rm sub}\equiv
 \Order (\epsilon_{\rm spin}^2)
 +\Order (\epsilon_{\rm r}) \,.
 \label{eq:subleading-symbol}
\end{equation}
\begin{proposition}[Linear-spin promotion]
\label{prop:spin-promotion}
For a minimally coupled massive spin-\(j\) source, the coefficient linear in
\(a^\mu\) is independent of \(j\). 
 Consequently it may be extracted from a
spin-\(\tfrac12\) amplitude and retained in the correlated limit
\eqref{eq:correlated-classical-limit}.  This statement does not determine the
quadratic or higher-spin coefficients.
\end{proposition}

\emph{Proof.}
We first isolate where the representation label \(j\) enters the three-point
on-shell amplitudes.
On the positive-helicity chiral branch, little-group covariance fixes
the minimal amplitude of two massive spin-\(j\) legs and one massless leg of
helicity \(h>0\) to be
\begin{equation}
 \M_3^{[j],h}
 =
 \left(\frac{\langle\mathbf{1}\mathbf3\rangle}{m_\Psi}\right)^{2j}
 \M_3^{[0],h}.
 \label{eq:minimal-spin-j-seed}
\end{equation}
Here,
\(\langle\mathbf{1}\mathbf3\rangle\) is the bold massive
\(\mathrm{SU}(2)\) little-group bracket, while \(\M_3^{[0],h}\) is the scalar
three-point amplitudes in Eqs.~\eqref{eq:3pt-seeds}--\eqref{eq:3pt-seeds-end}: \(\sqrt2\,q_\Psi m_\Psi x\) for a photon and
\(\kappa m_\Psi^2x^2/2\) for a graviton.  On the conjugate-helicity branch,
\(\langle\mathbf{1}\mathbf3\rangle\) and \(x\) are replaced by
\([\mathbf{1}\mathbf3]\) and \(\bar x\), respectively.  Thus all dependence on
the massive spin representation is confined to the bold-bracket factor.
Here, the
power of \(x\), which carries the massless little-group weight, is fixed by
the species and is independent of \(j\).

Now evaluate this factor in a diagonal spin-coherent state.  With the
classical spin vector \(a^\mu\) of Eq.~\eqref{eq:amu} and momentum transfer
\(\Delta^\mu \equiv (k-k')^\mu\), its leading-recoil matrix element is
\begin{align}
 \left(\frac{\langle\mathbf{1}\mathbf3\rangle}{m_\Psi}\right)^{2j}
 &\to \left(1+\sigma\frac{\Delta\cdot a}{2j}\right)^{2j}\notag\\
 &=1+\sigma\Delta\cdot a\notag\\
 &\quad+\frac{2j-1}{4j}(\Delta\cdot a)^2+\cdots ,
 \label{eq:finite-j-binomial}
\end{align}
where \(\sigma=\pm1\) labels the massless helicity convention
\cite{ArkaniHamed:2017jhn,ArkaniHamed:2019ymq,Moynihan:2019bor,
Chung:2019duq}.  
Reading \eqref{eq:finite-j-binomial} as a binomial expansion
gives three immediate facts; (i) the coefficient of \(\Delta\cdot a\) is
exactly one for every \(j\),
(ii) the quadratic and higher coefficients
depend on \(j\), with the coefficient of \((\Delta\cdot a)^2\) vanishing at
\(j=\tfrac12\) and approaching \(1/2\) as \(j\to\infty\), and (iii)
\(\lim_{j\to\infty}\left(1+\sigma\frac{\Delta\cdot a}{2j}\right)^{2j\sigma}
=\exp(\sigma\Delta\cdot a)\).

Finally, we consider the four-point tree-level on-shell amplitudes.
For the pure Dirac-QED Compton amplitude, there is no contact term. 
Local contact terms occur in the conversion amplitudes and gravitational amplitudes. 
We clarify that the pole residues determine the linear-spin coefficient, 
while the minimally coupled four-point contact terms are also included for the Ward-identity-required local completion.
Its residues on the two massive poles
factorize into products of the three-point seeds, so at one spin insertion
their coefficient inherits the representation-independent linear term in Eq.~\eqref{eq:finite-j-binomial}.  At the two-derivative order retained here, the
pole-free part is then the linear Ward-identity completion of those residues
once independent local higher-multipole couplings are excluded.  
The
factorization construction is summarized in Appendix~\ref{app:spinor-scalar},
and the full spin-\(\tfrac12\) tree-level completion is verified explicitly in Appendix~\ref{app:fermion}.  An anomalous Pauli interaction beyond the
Dirac \(g=2\) magnetic moment, as well as independent
curvature--field-strength and tidal operators,
is excluded by the
minimal-coupling assumption.  
Hence the representation-independent
three-point coefficient is also the coefficient linear in \(a^\mu\) of the
complete minimal four-point tree.  No analogous conclusion is drawn at
quadratic or higher order, whose coefficients are \(j\)-dependent by (ii).
\hfill\(\square\)\\

The Dirac fermion therefore serves as a coefficient extractor.
For bookkeeping, write
\begin{equation}
 \frac{\mathcal C_j}{\mathcal C_0}
 =1+c_{10}\epsilon_{\rm spin}
  +c_{01}\epsilon_{\rm r}
  +c_{20}\epsilon_{\rm spin}^2
  +c_{11}\epsilon_{\rm spin}\epsilon_{\rm r}+\cdots .
 \label{eq:double-expansion}
\end{equation}
Spin reversal removes every spin-even term before the no-recoil expansion:
\begin{equation}
 c_{10}
 =\lim_{\epsilon_{\rm r}\to0}
 \frac{1}{\epsilon_{\rm spin}}\,
 \frac{\mathcal C_j(a)-\mathcal C_j(-a)}
      {\mathcal C_j(a)+\mathcal C_j(-a)} \,.
 \label{eq:spin-coefficient-extractor}
\end{equation}
Let us consider that the direction of $a^\mu$ is set to be along the $z$-axis.
For \(j=\tfrac12\), \(a_z=\hbar/(2m_\Psi)\), so
\(\epsilon_{\rm spin}=\epsilon_{\rm r}/2\) in the aligned state; this
equality is precisely why the odd projection must precede the
limit.  Equation~\eqref{eq:spin-coefficient-extractor} reads \(c_{10}\);
only after that extraction is \(a_z\) replaced by a finite classical \(a\).
Let us define
\begin{align}
\begin{aligned}
\mathcal{M}_{\mathrm{av}} & \equiv \frac{\mathcal{M}\left(a_z\right)+\mathcal{M}\left(-a_z\right)}{2}\,, \\ 
\mathcal{M}_{\mathrm{spin}} &\equiv \frac{\mathcal{M}\left(a_z\right)-\mathcal{M}\left(-a_z\right)}{2}\,,
\end{aligned}
\end{align}
the operational order used below is therefore
\begin{align}
 \M_{\rm exact}^{[1/2]}
 &\xrightarrow{\ {\rm spin\ odd}\ }
 \frac{\M_{\rm spin}}{\omega a_z\M_{\rm av}}\,,\\
 c_{10} & = \lim_{\epsilon_{\rm r}\to0}
 \frac{\M_{\rm spin}}{\omega a_z\M_{\rm av}}\,,
% \frac{\M_{\rm spin}}{\omega a_z\M_{\rm av}}
% &\xrightarrow{\epsilon_{\rm r}\to0}c_{10},\\
\end{align}
followed by
\begin{align}
 %c_{10}&\to 
 \frac{\M_{\rm classical}^{(0)+(1)}}{\M_{\rm classical}^{(0)}}
 =1+c_{10}\omega a \,.
 \label{eq:order-of-operations}
\end{align}
Here \((n)\) counts powers of the classical spin, so that \(\M_{\rm classical}^{(0)}\) is the spinless amplitude of Sec.~\ref{sec:rn-tree} and \(\M_{\rm classical}^{(0)+(1)}\) is its extension to first order in \(a\). 
Because \(c_{10}\) is defined as a ratio, the common coupling and source-normalization factors cancel, and the substitution of the exterior parameters \eqref{eq:correlated-classical-limit} may be performed either before or after the extraction.

The calculation below takes both \(a^\mu\) and the incident momentum along
\(\hat z\).  We do not infer the full arbitrary-orientation Compton amplitude
by covariantizing this slice.  A vertical bar \(\left.\vphantom{X}\right|_\parallel\)
will denote evaluation in this aligned configuration.

The two QED diagrams give
\begin{multline}
 \M_{1/2}=q_\Psi^2\bar u(p')\left[
 \slashed{\e}_f^*
 \frac{\slashed p+\slashed k+m_\Psi}{2p\cdot k}
 \slashed{\e}_i\right.\\
 \left.
 -\slashed{\e}_i
 \frac{\slashed p-\slashed k'+m_\Psi}{2p\cdot k'}
 \slashed{\e}_f^*\right]u(p)\,.
 \label{eq:Dirac-Compton}
\end{multline}
Separating the spin-odd piece from recoil and then taking
\(\epsilon_{\rm r}\to0\) yields, for
\(a^\mu=(0,0,0,a)\) and \(k^\mu=\omega(1,0,0,1)\),
\begin{equation}
 \begin{aligned}
 \left.\mathcal C_{1/2}^{h_fh_i}\right|_{\parallel}
 &=\mathcal C_0^{h_fh_i}
 \left[\mathcal D_{fi}(a)+\delta_{\rm sub}\right]\,,
 \end{aligned}
 \label{eq:aligned-spin-dressing}
\end{equation}
and
\begin{equation}
 \begin{aligned}
 \mathcal D_{fi}(a)
 &\equiv1+\sigma_f(k-k')\cdot a\\
 &=1-\sigma_fa\omega(1-\cos\theta) \\
 &= 1-2 \sigma_fa\omega s_\theta^2\,,
 \end{aligned}
 \label{eq:aligned-dressing-factor}
\end{equation}
with $\sigma_f\equiv\sgn(h_f)$.
This covariant-looking form is asserted only on the aligned slice.
Appendix~\ref{app:fermion} gives the Dirac reduction.
After matching the all-incoming momentum convention and reversing the
helicity-index order, the aligned helicity-resolved \(g\to\gamma\) cross
sections of Ref.~\cite{Zheng:2026knmix} reproduce the angular and
final-photon-helicity dependence implied by \(\mathcal D_{fi}\). Our result is consistent up to convention-dependent signs with Ref.~\cite{Zheng:2026knmix}, although we do not attempt a detailed sign-by-sign comparison of the amplitudes because of differences in the polarization and Levi--Civita conventions.
\begin{proposition}[Aligned factorization]
\label{prop:linear-spin-factorization}
For a minimally coupled Dirac source, the complete tree species matrix
includes the Born, contact, and graviton-pole graphs.  In the aligned
no-recoil limit,
it has exactly one spin-dependent double-copy:
\begin{equation}
 \left.{\mathbf M}^{[1/2]}\right|_\parallel
 =
 \left[
 \mathcal D_{fi}(a)
 +\delta_{\rm sub}
 \right]
 \left.{\mathbf M}^{[0]}\right|_\parallel .
\end{equation}
\end{proposition}
\vspace{1\baselineskip}

After dividing the fermion kernels by the spin-conserving overlap
\(\bar u u/(2m_\Psi)\), the stress-tensor contraction in
Appendix~\ref{app:fermion} gives the same aligned dressing for the graviton
pole.  Full-tree-level factorization then places exactly one spin-dependent gauge
kernel in every species block.  In particular,
\begin{equation}
 \begin{aligned}
 R\,\mathcal C_{1/2}\,\mathcal C_0
 &=R\,\mathcal C_0^2
 \left[\mathcal D_{fi}+\delta_{\rm sub}\right],\\
 R\,\mathcal C_{1/2}^2
 &=R\,\mathcal C_0^2
 \left[\mathcal D_{fi}^2+\delta_{\rm sub}\right]\\
 &=R\,\mathcal C_0^2
 \bigl[1+2\sigma_f(k-k')\cdot a
 +\delta_{\rm sub}\bigr].
 \end{aligned}
\end{equation}
Thus \(R\,\mathcal C_{1/2}^2\) would double the coefficient linear in \(a\);
the physical gravitational block is
\(R\,\mathcal C_{1/2}\,\mathcal C_0\).  Appendix~\ref{app:fermion} gives the
stress-tensor contraction and standard factorization identity.
This result is consistent with the factorization formulas in Refs.~\cite{Choi:1994ax,Holstein:2006bh}. 
%(See Eqs.(3.35)--(3.36), (3.56) and Eqs. (27)--(28), (38)--(44).)

Choose the incident momentum and spin along \(+\hat z\).  At leading recoil
order,
\begin{equation}
 (k-k')\cdot a=-a\omega(1-\cos\theta)
 =- 2 a\omega s_\theta^2 .
 \label{eq:aligned}
\end{equation}
For final positive helicity, \(\sigma_f=+1\), so every correction is the
corresponding spinless wave amplitude multiplied by the common factor
\(-2a\omega\sth^2\), which is explicitly 
\begin{align}
 \label{eq:delta-f}
 \delta f_{\gamma\gamma}^{++}
 &=-2a\omega M\cth^2
 +2a\omega\frac{Q^2}{M}\cth^2\sth^2\,,\\
 \delta f_{\gamma\gamma}^{+-}
 &=+2a\omega\frac{Q^2}{M}\sth^4\,,\\
 \delta f_{gg}^{++}
 &=-2a\omega M\cth^4\,,\\
 \delta f_{gg}^{+-}
 &=-2a\omega M\sth^4\,,\\
 \delta f_{\gamma g}^{++}
 &=-2a\omega Q\cth^3\sth\,,\\
 \delta f_{\gamma g}^{+-}
 &=+2a\omega Q\cth\sth^3 \,.
 \label{eq:delta-f-end}
\end{align}
The other helicity row follows from simultaneous helicity and spin reversal.
These functions have finite Wigner content.

\begin{proposition}[Aligned \(J\)-diagonal statement]
\label{prop:aligned}
For the kinematics \eqref{eq:aligned}, the term linear in \(a\) in every
fixed-\(J\) diagonal Jacob--Wick integral vanishes for \(J\geq3\).  Thus this
restricted slice supplies no linear-spin correction, and hence no projector
test, in those diagonal blocks.
\end{proposition}

The mechanism is Jacobi orthogonality.  Up to nonzero normalization factors,
\begin{equation}
\begin{gathered}
 d^J_{22}=\cth^4P_{J-2}^{(0,4)}\,,\quad
 d^J_{2,-2}=\sth^4P_{J-2}^{(4,0)}\,,\\
 d^J_{12}\propto\cth^3\sth P_{J-2}^{(1,3)}\,,\quad
 d^J_{1,-2}\propto\cth\sth^3P_{J-2}^{(3,1)}\,,\\
 d^J_{11}=\cth^2P_{J-1}^{(0,2)}\,,\quad
 d^J_{1,-1}=\sth^2P_{J-1}^{(2,0)} \,.
\end{gathered}
\end{equation}
For the graviton and conversion channels, multiplying by
\eqref{eq:delta-f}--\eqref{eq:delta-f-end} leaves the Jacobi weight times a constant, so
orthogonality to \(P_0\) kills \(J\geq3\).  For the photon channel the
remaining polynomial has degree at most one, so orthogonality kills
\(P_{J-1}\) for \(J-1\geq2\).  The surviving \(J=2\) integrals are recorded in
Appendix~\ref{app:aligned}.  They are analytic functions of \(t\) in this
special slice and depend on how helicities with different \(m\) are assembled.

Appendix~\ref{app:aligned} also shows that the two formal \(J=2\) matrices
obey
\begin{equation}
 \comm{\K_2^+}{(\delta\A_{2,+}^{\rm formal})^{\TF}}\neq0\,,
 \quad
 \comm{\K_2^-}{(\delta\A_{2,-}^{\rm formal})^{\TF}}\neq0 \,.
 \label{eq:J2-noncommuting-summary}
\end{equation}
This is the direct difference from the RN case:
the RN tree-level matrix elements commute with
\(\K_J^P\) for every radiative \(J\), whereas the formal linear-spin
\(J=2\) matrix does not.

The axial mismatch comes entirely from the spin-dressed Thomson term:
multiplication by \((1-\cos\theta)\) leaks the spinless \(J=1\) polynomial into
\(J=2\), a sector absent from the spinless \(\ell\geq2\) theorem.
Appendix~\ref{app:aligned} gives the individual cancellations.

The nonzero commutators show only that the unchanged RN projectors fail on
this formal aligned \(J=2\) construction.  Its photon and graviton columns do
not form a common \((J,m)\) sector, while rotation couples neighboring
\(J\)'s and may require an \(a\)-dependent transformation.  This is therefore
not a KN separability test; that test requires the complete
fixed-\(m\) angular system and the curved-background equations.

\section{Discussion and conclusions}
\label{sec:discussion}
In this work, we investigated the extent to which photon-graviton mixing and its decoupling structure in charged-black-hole perturbation theory are encoded in flat-space on-shell scattering amplitudes. 
As part of this analysis, we computed the tree-level four-point amplitudes in Einstein--Maxwell theory minimally coupled to a heavy charged complex scalar using the spinor-helicity formalism. 
We organized elastic photon scattering, elastic graviton scattering, and photon-graviton conversion into a single species-space amplitude matrix and obtained results consistent with known Compton amplitudes, gravitational factorization relations, and low-frequency scattering cross sections.
In particular, we showed that a candidate constant channel transformation separating the electromagnetic and gravitational perturbation modes of a Reissner--Nordstr\"{o}m (RN) background can be extracted from the partial-wave matrix reconstructed from photon and graviton scattering off a heavy charged source. From this perspective, instead of searching directly for a change of variables in the differential equations, one asks whether the partial-wave amplitude and the curved-background potential possess the same spectral projectors. The mode-decoupling problem can thereby be formulated as an algebraic problem in channel space.
In the fixed-source limit, corresponding to the zeroth order of the recoil expansion, the trace-free part of the partial-wave matrix obtained from the flat-space tree amplitudes is proportional to the trace-free part of the  Moncrief coupling matrix \(\K_\ell^P\) appearing in the RN perturbation equations. Consequently, for every radiative multipole \(\ell\geq2\) and both parity sectors \(P=\pm1 \), we find
\begin{equation}
[\A_{\ell,0}^P,\K_\ell^P]=0 \,.
\end{equation}
The channel eigenspaces selected by the on-shell amplitudes therefore coincide with those defining the Moncrief variables. The corresponding mixing angle also agrees with the Moncrief angle, up to the relative phase convention for the two channels and an interchange of the eigenvectors. Thus, the channel direction characterizing the decoupling of RN perturbations can be identified from flat-space long-range scattering data before directly solving the curved-background perturbation equations.

The amount of information determined by the tree-level amplitudes must, however, be stated carefully. Through the first-Born correspondence, the amplitudes determine the channel-dependent, trace-free part of the leading \(r^{-3}\) term in the weak-field potential. They do not by themselves determine the channel-independent contribution, the higher-order terms in the \(1/r\) expansion, or the complete RN potential at finite radius. Higher-order coefficients in the weak-field expansion are related to higher post-Minkowskian information and, on the amplitude side, generally also involve the classical parts of loop amplitudes. Nevertheless, a finite number of loop orders need not suffice to reconstruct a potential valid at all radii. Such a reconstruction may require a resummation of infinitely many weak-field coefficients or additional nonperturbative input.
For this reason, in Sec.~\ref{sec:radial} we introduced the complete classical RN perturbation potentials as independent curved-background input. For each fixed \((\ell,P)\), we showed that their channel matrices close at every radius in the two-dimensional algebra
\begin{equation}
 \mathbf V_{\ell,s}(r)
 \in
 \Span\left\{
 \mathbf{1}_2,\K_\ell^P
 \right\}\,.
\end{equation}
This closure of the algebra implies that the constant projectors identified from the amplitudes diagonalize the complete RN radial system throughout the full radial domain. Our conclusion is therefore not that complete RN separability follows from flat-space amplitudes alone. Rather, the amplitudes select candidate projectors in the weak-field regime, while the algebra of the exact RN potentials independently establishes that the same projectors remain valid at every radius. The argument thus consists of two logically distinct steps.

Because the Moncrief variables separating the RN system are already known, the practical advantage of this method may not be immediately apparent in the RN case. For coupled wave systems in which the appropriate transformation is unknown, however, the present approach may provide a useful way to infer candidate channel projectors and decoupling variables directly from scattering amplitudes.

A natural example is the Kerr--Newman (KN) system, which carries both angular momentum and electric charge.
As a step toward a rotating charged source, we used a minimally coupled Dirac field as an extractor of the linear-in-spin coefficient. By first separating the spin-odd contribution from the spin-even recoil terms through spin reversal and only then taking the zero-recoil limit, we extracted the spin-representation-independent linear-spin coefficient associated with the \(g=2\) magnetic-dipole interaction. On the aligned kinematic slice, every entry of the complete tree-level species-space amplitude matrix factorizes into the corresponding scalar amplitude multiplied by a single common factor. This result allows the linear-spin term extracted from the Dirac amplitude to be promoted to a minimally coupled source carrying finite classical spin. It does not, however, determine the coefficients at quadratic or higher orders in spin.
A formal \(J=2\) projection on the aligned kinematic slice produces a linear-spin partial-wave matrix that does not commute with the RN Moncrief coupling matrix \(\K_2^P\). These statements apply only to the aligned kinematic slice and do not constitute a test of separability in KN spacetime. In a rotating background, different values of \(J\) can couple within a sector of fixed azimuthal number \(m\), and the required channel transformation may itself depend on the spin parameter. The noncommutativity of the formal \(J=2\) matrix therefore shows only that the unchanged RN projectors cannot be applied directly to the rotating system.

Our results demonstrate that on-shell amplitudes can efficiently identify the channel structure and candidate decoupling variables of a coupled wave system without reconstructing its complete curved-background equations. On the amplitude side, no information specific to the black-hole horizon is used directly. Instead, the mass and charge of the heavy source are matched to the asymptotic monopole data through Eq.~\eqref{eq:MQ}. In this sense, the extraction of channel projectors from amplitudes may be applicable beyond black-holes to more general long-range scattering systems with multiple asymptotic channels. Whether the candidate projectors remain valid at finite distances must nevertheless be tested independently using the complete wave operator of the system under consideration, just as in the RN analysis presented here.

An important direction for future work is to determine whether candidate decoupling variables can be constructed systematically from scattering amplitudes for general coupled wave systems, including systems unrelated to black-holes. For the KN system, it will be necessary to construct the complete angular-mode system coupled at fixed-\(m\) from the amplitude side and investigate its decoupling structure. The calculation of amplitude matrices through quadratic and higher orders in spin, together with a direct comparison with the KN perturbation equations, will also be important. For the RN system, it would be interesting to include the classical conservative contributions contained in the one-loop amplitudes and determine whether their trace-free partial-wave operators continue to lie in \(\Span\{\mathbf 1_2,\K_\ell^P\}\), or how their commutators with \(\K_\ell^P\) are modified. Such studies may clarify more generally the relation between algebraic structures appearing in scattering amplitudes and separability in coupled wave systems, including those arising in rotating black-hole backgrounds.

\acknowledgments

The authors would like to express their sincere gratitude to Katsuki Aoki for his important insights and many fruitful discussions.
This project began in earnest following his seminar at Chiba University and grew out of an attempt to extend the framework for describing black-hole microstates developed in Ref.~\cite{Akpinar:2026oni} to elastic coupled-channel scattering on a charged black-hole background.
His broad expertise in black-hole scattering amplitudes and the extensive discussions that followed the seminar played an important role in shaping this work.
In particular, these efforts led directly to the central question addressed here: whether flat-space on-shell amplitudes already encode the channel structure of black-hole perturbations.
K.T.\ would also like to express his deepest gratitude to T.K.\ for welcoming me into the Chiba University Particle Theory Group at an early stage of my studies,
introducing him to scattering amplitudes as a research subject, and teaching me the foundations of particle theory through seminars and numerous discussions. 
%K.T\ is especially grateful for T.K.'s continued guidance throughout this project, including his careful comments on the manuscript down to its finer details. 
K.T.\ also thank the members of the Chiba Particle Physics Group and gratefully acknowledge the Center for Frontier Science at Chiba University for providing an excellent academic environment and financial support.
The work of T.K.\ is supported by the JSPS Grant-in-Aid for Scientific Research Grant No.\,22K21347 and 25K07276.

%\clearpage
\onecolumngrid
\appendix

\section{Spinor-helicity construction of the scalar tree matrix}
\label{app:spinor-scalar}

\subsection{Conventions and three-point seeds}

We label the external legs according to the physical process $ 1(p)+2(k)\to 3(p')+4(k'),$
but perform the spinor-helicity construction in this appendix using
the all-outgoing convention.  The all-outgoing momenta are therefore assigned as
\begin{equation}
 p_1=-p,\qquad
 k_2=-k,\qquad
 p_3=p',\qquad
 k_4=k'.
\end{equation}
The prescription for returning to the physical in-out convention,
including the crossing of the helicity on leg $2$, is given after
Eq.~\eqref{eq:det-identity}.  All amplitudes displayed in the main text are already
written in the physical convention, with legs $1$ and $2$ incoming
and legs $3$ and $4$ outgoing.
Accordingly, momentum conservation and the on-shell conditions read
\begin{equation}
 p_1+k_2+p_3+k_4=0,\qquad
 p_1^2=p_3^2=m_\Phi^2,\qquad k_2^2=k_4^2=0,
 \label{eq:all-outgoing}
\end{equation}
and define
\begin{equation}
 D_2=\langle2|\boldsymbol{p}_1|2],\qquad
 D_4=\langle4|\boldsymbol{p}_1|4],\qquad
 s_{24}=\langle24\rangle[42].
 \label{eq:D34}
\end{equation}
We follow the conventions for the massive spinor-helicity formalism of Refs.~\cite{ArkaniHamed:2017jhn,Durieux:2019eor}.
Throughout this appendix $m\equiv m_\Phi$ and $q\equiv q_\Phi$. Squaring \(p_3=-p_1-k_2-k_4\) immediately gives
\begin{equation}
 D_2+D_4+s_{24}=0.
 \label{eq:Didentity}
\end{equation}
Our standard polarizations are
\begin{equation}
 \e_+^\mu(k;\xi)=\frac{\langle \xi|\sigma^\mu|k]}
 {\sqrt2\langle \xi k\rangle},\qquad
 \e_-^\mu(k;r)=-\frac{[\xi|\bar{\sigma}^\mu|k\rangle}
 {\sqrt2[\xi k]} ,
 \label{eq:polarizations}
\end{equation}
where $|\xi], |\xi\rangle$ are reference spinors.
The \(\sqrt2\) factors below follow from this normalization; little-group
covariance fixes the structures but not the coupling normalization. 

On three-point kinematics, introduce the standard massive
spinor-helicity variables \cite{ArkaniHamed:2017jhn}
\begin{equation}
 x=\frac{\langle\xi|\boldsymbol{p}_1|2]}{m\langle\xi2\rangle},
 \qquad
 \bar x=\frac{[\xi|\boldsymbol{p}_1|2\rangle}{m[\xi2]} .
 \label{eq:xvar}
\end{equation}
Minimal SQED and gravity give
\begin{align}
 \label{eq:3pt-seeds}
 \M_3(\Phi\bar\Phi\gamma^+)&=\sqrt2\,qm x,
 &\M_3(\Phi\bar\Phi\gamma^-)&=\sqrt2\,qm\bar x,\\
 \M_3(\Phi\bar\Phi g^+)&=\frac{\kappa m^2}{2}x^2,
 &\M_3(\Phi\bar\Phi g^-)&=\frac{\kappa m^2}{2}\bar x^2 .
 \label{eq:3pt-seeds-end}
\end{align}

\subsection{The gauge kernel and local completion}

The two internal scalar momenta may be chosen as
\(P_2=p_1+k_2\) and \(P_4=p_1+k_4\).  Their inverse propagators are
\begin{equation}
 P_2^2-m^2=D_2,\qquad P_4^2-m^2=D_4.
\end{equation}
Suppressing the common \(S\)-matrix phase, tree factorization therefore
requires
\begin{align}
 \left.D_2\,q^2\mathcal C_0^{h_2h_4}\right|_{D_2=0}
 &=
 \M_3(p_1,2^{h_2},-P_2)\,
 \M_3(P_2,4^{h_4},p_3),\\
 \left.D_4\,q^2\mathcal C_0^{h_2h_4}\right|_{D_4=0}
 &=
 \M_3(p_1,4^{h_4},-P_4)\,
 \M_3(P_4,2^{h_2},p_3).
\end{align}
The helicity superscripts in what follows are ordered as $(h_2,h_4)$.
Substitution of the seeds in \eqref{eq:3pt-seeds}--\eqref{eq:3pt-seeds-end}, followed by elimination of
the reference spinor \(\xi\) on each pole, fixes the little-group covariant
numerators
\begin{align}
 \mathcal N^{++}&=-2m^2[24]^2,&
 \mathcal N^{+-}&=2\langle4|\boldsymbol{p}_1|2]^2,\\
 \mathcal N^{-+}&=2\langle2|\boldsymbol{p}_1|4]^2,&
 \mathcal N^{--}&=-2m^2\langle24\rangle^2.
\end{align}
Explicitly, the two residues are
\begin{equation}
 \left.D_2\mathcal C_0^{h_2h_4}\right|_{D_2=0}
 =\left.\frac{\mathcal N^{h_2h_4}}{D_4}\right|_{D_2=0},
 \qquad
 \left.D_4\mathcal C_0^{h_2h_4}\right|_{D_4=0}
 =\left.\frac{\mathcal N^{h_2h_4}}{D_2}\right|_{D_4=0}.
\end{equation}
Factorization determines these residues but, by itself, allows a pole-free
polynomial.  
At the two-derivative order the required polynomial is the
SQED seagull.  After crossing to the physical process, the three QED
diagrams in Fig.~\ref{fig:gamma-tree} read
\begin{align}
 \M_s&=2q^2
 \frac{(\e_i\cdot p)(\e_f^*\cdot p')}{p\cdot k},&
 \M_u&=-2q^2
 \frac{(\e_i\cdot p')(\e_f^*\cdot p)}{p\cdot k'},\\
 \M_{\rm ct}&=-2q^2\,\e_i\cdot\e_f^* .
\end{align}
The coefficient of the contact term is fixed, rather than optional.  Indeed,
under \(\e_i\to k\), on-shell momentum conservation gives
\begin{equation}
 \left.(\M_s+\M_u)\right|_{\e_i\to k}
 =2q^2\,k\cdot\e_f^*,\qquad
 \left.\M_{\rm ct}\right|_{\e_i\to k}
 =-2q^2\,k\cdot\e_f^*,
\end{equation}
and the outgoing Ward identity works in the same way.  
Combining the
gauge-invariant expression over the common denominator \(D_2D_4\) gives the
four helicity components
\begin{align}
 \label{eq:C-spinor}
 \mathcal C_0^{++}
 &=-\frac{2m^2[24]^2}{D_2D_4},&
 \mathcal C_0^{+-}
 &=\frac{2\langle4|\boldsymbol{p}_1|2]^2}{D_2D_4},\\
 \mathcal C_0^{-+}
 &=\frac{2\langle2|\boldsymbol{p}_1|4]^2}{D_2D_4},&
 \mathcal C_0^{--}
 &=-\frac{2m^2\langle24\rangle^2}{D_2D_4}.
 \label{eq:C-spinor-end}
\end{align}
The seagull is not an extra optional term: it is the polynomial required to
complete the pole ansatz into \eqref{eq:C-spinor}--\eqref{eq:C-spinor-end}.  A compact two-step gluing
proof uses
\begin{equation}
 \langle4|\boldsymbol{p}_1|2]\langle2|\boldsymbol{p}_1|4]
 =D_2D_4+m^2s_{24},
 \label{eq:det-identity}
\end{equation}
together with \eqref{eq:Didentity}. 
Apply the crossing map
$p_1=-p,~
 k_2=-k,~
 p_3=p',~
 k_4=k'$
together with
$|-k\rangle=i|k\rangle$ and $|-k]=i|k]$.
Then, $h_2=-h_i, h_4=h_f.$
Thus physical $(h_f,h_i)=(+,+)$ comes from the all-outgoing
component $(h_2,h_4)=(-,+)$, whereas physical $(+,-)$ comes
from $(+,+)$. After stripping the standard external helicity phases, the
needed numerator maps are
\begin{equation}
 \langle2|\boldsymbol{p}_1|4]^2\to \Xi_\Phi,\qquad
 [24]^2\to  t.
\end{equation}
This displays explicitly how the seagull contribution contained in
\(\Xi_\Phi=m^2t-D_sD_u\) is present in the compact spinor expression.

Factorization fixes the pole residues but does not determine local terms
without poles.  The leading such terms appear at dimension six in the
same-helicity sector and at dimension eight in the opposite-helicity
sector.  Their coefficients are independent Wilson coefficients and are
not fixed by the three-point amplitudes.  Throughout this work,
minimal means that these additional local terms are set to zero.
Higher-derivative operators may modify the contact contributions to
individual partial waves, but not the pole residues fixed by
factorization.

\subsection{Gravitational dressing}
The minimal \(\gamma\gamma\) graviton pole has
\begin{align}
 \label{eq:h-pole-spinor}
 \left(\M_{\gamma\gamma}^{(g{\rm -pole})}\right)^{++}
 &=\left(\M_{\gamma\gamma}^{(g{\rm -pole})}\right)^{--}=0,\\
 \left(\M_{\gamma\gamma}^{(g{\rm -pole})}\right)^{+-}
 &=-\frac{\kappa^2}{4}
 \frac{\langle4|\boldsymbol{p}_1|2]^2}{s_{24}},&
 \left(\M_{\gamma\gamma}^{(g{\rm -pole})}\right)^{-+}
 &=-\frac{\kappa^2}{4}
 \frac{\langle2|\boldsymbol{p}_1|4]^2}{s_{24}} .
 \label{eq:h-pole-spinor-end}
\end{align}
The zero in the first line is a statement about the two-derivative theory;
operators such as \(RF^2\) may generate same-helicity contact terms.

The gauge-invariant half-dressings can be written
\begin{align}
 \label{eq:N-spinor}
 N_4^+&=-\frac{\langle2|\boldsymbol{p}_1|4]}{\sqrt2\langle24\rangle},&
 N_4^-&=\frac{[2|\boldsymbol{p}_1|4\rangle}{\sqrt2[24]},\\
 N_2^+&=-\frac{\langle4|\boldsymbol{p}_1|2]}{\sqrt2\langle24\rangle},&
 N_2^-&=\frac{[4|\boldsymbol{p}_1|2\rangle}{\sqrt2[24]} .
 \label{eq:N-spinor-end}
\end{align}
The Kawai--Lewellen--Tye (KLT) kernel is
\begin{equation}
 R=\frac{D_2D_4}{2s_{24}} .
 \label{eq:R-spinor}
\end{equation}
Thus the all-outgoing species matrix is
\begin{equation}
 {\mathbf M}_4=
 \begin{pmatrix}
 q^2\mathcal C_0+\M_{\gamma\gamma,g}
 &\dfrac{\kappa q}{2}N_4\mathcal C_0\\[2mm]
 \dfrac{\kappa q}{2}N_2\mathcal C_0
 &\dfrac{\kappa^2}{8}R\mathcal C_0^2
 \end{pmatrix}.
 \label{eq:M4-spinor}
\end{equation}
Here the species labels of the all-outgoing amplitudes refer to legs (2) and (4), respectively. Upon crossing leg (2) to the initial state, these become the initial and final species labels, \(s_2\to s_i\) and \(s_4\to s_f\). The holomorphic weight of a given graviton leg selects one half-dressing.
Using both \(N_2N_4\) and \(R\) would double-dress the amplitude and spoil
factorization.

Equations~\eqref{eq:all-outgoing}--\eqref{eq:M4-spinor} are exact at finite
\(m\) and include recoil; no expansion in \(\omega/m\) has been made here.
The replacement \eqref{eq:MQ} is applied only after crossing and wave
normalization.

\section{Projection details and the radial Born matching}
\label{app:projection}
\subsection{Fixed-source projection}
Let \(I,J\in\{\gamma,g\}\) label the initial and final species, respectively. The parity-projected partial-wave matrix in the helicity convention is
\begin{equation}
 \left(\A_{\ell}^{P,\mathrm{hel}}\right)_{JI}
 =\int_{-1}^{1}\dd z\,
\left[
d^\ell_{s_J,s_I}f^{JI}_{+s_J,+s_I}
+P\eta_I d^\ell{s_J,-s_I}f^{JI}_{+s_J,-s_I} .
\right]
 \label{eq:parity-projection-detail}
\end{equation}
Following the rephasing introduced after
Eq.~\eqref{eq:parity-state}, we express this matrix as
\begin{equation}
 \A_\ell^P
 =\mathbf U_P\A_{\ell}^{P,\mathrm{hel}}\mathbf U_P,
 \qquad
 \mathbf U_P\equiv\operatorname{diag}(P,1),
\end{equation}
where the channels are ordered as \((\gamma,g)\). Time-reversal invariance implies that the matrix is symmetric.
\begin{equation}
    \left(\A_\ell^P\right)^{\mathsf T}=\A_\ell^P
\end{equation}
At fixed \(\ell\), we suppress the \(\ell\) label on the matrix elements and coefficients. We parameterize the two independent components of
\(\left(\A_\ell^P\right)^{\TF}\) by \(B_P\) and \(D_P\)
\begin{equation}
 A_{\gamma g}^P=A_{g\gamma}^P\equiv P Q\,B_P,
 \qquad
 A_{\gamma\gamma}^P-A_{gg}^P\equiv M\,D_P.
 \label{eq:BD-definitions}
\end{equation}
It follows algebraically that
\begin{equation}
 \left(\A_\ell^P\right)^{\TF}
 =\begin{pmatrix}
 MD_P/2&PQB_P\\
 PQB_P&-MD_P/2
 \end{pmatrix}.
 \label{eq:BD-matrix}
\end{equation}
Equation~\eqref{eq:BD-definitions} defines \(B_P\) and \(D_P\), while
Eq.~\eqref{eq:BD-matrix} follows directly from that definition. The
Wigner-\(d\) integrals below determine these coefficients and establish the nontrivial relation between them.

We use the following convention for the Wigner \(d\) functions:
\begin{equation}
 \begin{aligned}
 d^\ell_{m'm}(\theta)
 &=
 \sqrt{
  (\ell+m)!(\ell-m)!
  (\ell+m')!(\ell-m')!
 }
 \\
 &\quad\times
 \sum_r
 \frac{
  (-1)^{m'-m+r}
  \cth^{2\ell+m-m'-2r}
  \sth^{m'-m+2r}
 }{
  (\ell+m-r)!\,r!\,
  (m'-m+r)!\,
  (\ell-m'-r)!
 } .
 \end{aligned}
 \label{eq:wigner-finite}
\end{equation}
where only terms with nonnegative factorials are retained. Let \( z\equiv\cos\theta,~u=\cth^2=(1+z)/2\),~ Then
\begin{equation}
 \int_{-1}^{1}\dd z\,\cth^{2\alpha}\sth^{2\beta}
 =2\int_0^1\dd u\,u^\alpha(1-u)^\beta
 =2\mathrm B(\alpha+1,\beta+1),
 \label{eq:beta-reduction}
\end{equation}
where \(\mathrm B\) denotes the Euler beta function.
For the photon-graviton conversion channel, we factor out the common
charge \(Q\) from
Eqs.~\eqref{eq:rn-wave}--\eqref{eq:rn-wave-end} and define the
helicity-nonflip and helicity-flip integrals by
\begin{align}
 I_{\rm nonflip}
 &=\int_{-1}^{1}\dd z\,
 d^\ell_{1,2}(\theta)\frac{\cth^3}{\sth},
 \label{eq:I-definitions}\\
 I_{\rm flip}
 &=-\int_{-1}^{1}\dd z\,
 d^\ell_{1,-2}(\theta)\cth\sth,\\
 B_P&=I_{\rm nonflip}+P I_{\rm flip},
 \qquad P=\pm1.
 \label{eq:I-definitions-end}
\end{align}
Let
\(\mathcal N_{\ell,12}\equiv
\sqrt{(\ell+2)!(\ell-2)!(\ell+1)!(\ell-1)!}\).
Substituting Eq.~\eqref{eq:wigner-finite} into the two integrals above
and applying Eq.~\eqref{eq:beta-reduction} term by term, we obtain
\begin{align}
 I_{\rm nonflip}
 &=2\mathcal N_{\ell,12}
 \sum_{r=1}^{\ell-1}
 \frac{(-1)^{r-1}\mathrm B(\ell+3-r,r)}
 {(\ell+2-r)!\,r!\,(r-1)!\,(\ell-1-r)!},
 \label{eq:I-beta-sums}\\
 I_{\rm flip}
 &=2\mathcal N_{\ell,12}
 \sum_{r=0}^{\ell-2}
 \frac{(-1)^r\mathrm B(\ell-r,r+3)}
 {(\ell-2-r)!\,r!\,(r+3)!\,(\ell-1-r)!}.
 \label{eq:I-beta-sums-end}
\end{align}
Using \(\mathrm B(p,q)=(p-1)!(q-1)!/(p+q-1)!\) and evaluating the finite sums, we find
\begin{equation}
 I_{\rm nonflip}=\frac{2}{\sqrt\Lambda},
 \qquad
 I_{\rm flip}
 =\frac{4}{(\Lambda+2)\sqrt\Lambda}.
 \label{eq:I-closed}
\end{equation}
Therefore,
\begin{equation}
 B_+
 =\frac{2(\Lambda+4)}
 {(\Lambda+2)\sqrt\Lambda},
 \qquad
 B_-
 =\frac{2\Lambda}{(\Lambda+2)\sqrt\Lambda}
 =\frac{2\sqrt\Lambda}{\Lambda+2}.
 \label{eq:B-closed}
\end{equation}
For the diagonal channels, we use the Jacobi-polynomial representations
\begin{equation}
 \begin{aligned}
 d^\ell_{11}(\theta)
 &=\cth^2P_{\ell-1}^{(0,2)}(z),\\
 d^\ell_{22}(\theta)
 &=\cth^4P_{\ell-2}^{(0,4)}(z),\\
 d^\ell_{2,-2}(\theta)
 &=\sth^4P_{\ell-2}^{(4,0)}(z).
 \end{aligned}
 \label{eq:diagonal-Jacobi}
\end{equation}
Combining these expressions with the angular factors in
Eqs.~\eqref{eq:rn-wave}--\eqref{eq:rn-wave-end}, we define
\begin{align}
 \mathcal J_\gamma(\delta)
 &=\int_{-1}^{1-\delta}\dd z\,
 \frac{(1+z)^2}{2(1-z)}
 P_{\ell-1}^{(0,2)}(z),
 \label{eq:J-integrals}\\
 \mathcal J_{g,\rm nonflip}(\delta)
 &=\int_{-1}^{1-\delta}\dd z\,
 \frac{(1+z)^4}{8(1-z)}
 P_{\ell-2}^{(0,4)}(z),\\
 \mathcal J_{g,\rm flip}
 &=\int_{-1}^{1}\dd z\,
 \frac{(1-z)^3}{8}
 P_{\ell-2}^{(4,0)}(z),
 \label{eq:J-integrals-end}
\end{align}
where \(\delta>0\) regulates the forward singularity at \(z=1\).

As \(z\to1\), the first two integrands have the same singular behavior, \(2/(1-z)\). Regulating both diagonal channels with the same cutoff \(z\leq1-\delta\), we obtain
\begin{align}
 \mathcal J_\gamma(\delta)
 &=2\log\frac{2}{\delta}
 +\int_{-1}^{1}\dd z
 \left[
 \frac{(1+z)^2P_{\ell-1}^{(0,2)}(z)}
 {2(1-z)}
 -\frac{2}{1-z}
 \right]
 +\Order(\delta),
 \label{eq:J-subtractions}\\
 \mathcal J_{g,\rm nonflip}(\delta)
 &=2\log\frac{2}{\delta}
 +\int_{-1}^{1}\dd z
 \left[
 \frac{(1+z)^4P_{\ell-2}^{(0,4)}(z)}
 {8(1-z)}
 -\frac{2}{1-z}
 \right]
 +\Order(\delta).
 \label{eq:J-subtractions-end}
\end{align}
Using
\begin{equation}
 P_n^{(\alpha,\beta)}(z)
 =
 \frac{(\alpha+1)_n}{n!}
 {}_2F_1\left(
 -n,n+\alpha+\beta+1;
 \alpha+1;\frac{1-z}{2}
 \right),
\end{equation}
where \({}_2F_1\) denotes the Gauss hypergeometric function and
\((a)_n\) denotes the Pochhammer symbol. Since \(n\) is a
nonnegative integer, the hypergeometric series terminates, yielding
a finite-sum representation of each Jacobi polynomial. After setting \(u=(1+z)/2\), the nonconstant
terms can be integrated term by term using
Eq.~\eqref{eq:beta-reduction}. The constant term contains the common
forward singularity and gives \(2\log(2/\delta)\). Evaluating the
remaining finite sums, we obtain
\begin{align}
 \mathcal J_\gamma(\delta)
 &=2\log\frac{2}{\delta}
 -4H_\ell+\frac{2}{\Lambda+2}
 +\Order(\delta),
 \label{eq:intermediate-proj}\\
 \mathcal J_{g,\rm nonflip}(\delta)
 &=\mathcal J_\gamma(\delta)
 +\frac{6}{\Lambda}
 +\Order(\delta),\\
 \mathcal J_{g,\rm flip}
 &=\frac{12}{(\Lambda+2)\Lambda},
 \label{eq:intermediate-proj-end}
\end{align}
where \(H_\ell=\sum_{r=1}^{\ell}r^{-1}\). The diagonal coefficient defined in Eq.~\eqref{eq:BD-definitions} is therefore
\begin{equation}
 D_P
 =\lim_{\delta\to0}
 \left[
 \mathcal J_\gamma(\delta)
 -\mathcal J_{g,\rm nonflip}(\delta)
 -P\mathcal J_{g,\rm flip}
 \right].
 \label{eq:D-definitions}
\end{equation}
We consequently find
\begin{align}
 D_+
 &=-\frac{6}{\Lambda}
 -\frac{12}{(\Lambda+2)\Lambda}
 =-\frac{6(\Lambda+4)}
 {(\Lambda+2)\Lambda},
 \label{eq:D-closed}\\
 D_-
 &=-\frac{6}{\Lambda}
 +\frac{12}{(\Lambda+2)\Lambda}
 =-\frac{6}{\Lambda+2}.
 \label{eq:D-closed-end}
\end{align}
The photon and graviton helicity-nonflip integrals contain the same
forward term \(2\log(2/\delta)\). This term therefore cancels in the
difference defining \(D_P\). At the matrix level, it appears with the
same coefficient in the \(\gamma\gamma\) and \(gg\) entries and is thus
proportional to \(\mathbf{1}_2\). It consequently drops out of
\(\left(\A_\ell^P\right)^{\TF}\).

The parity-dependent coefficients can be written uniformly as
\begin{equation}
 B_P
 =
 \frac{2(\Lambda+2+2P)}
 {(\Lambda+2)\sqrt{\Lambda}},
 \qquad
 D_P
 =
 -\frac{6(\Lambda+2+2P)}
 {(\Lambda+2)\Lambda}.
\end{equation}
Hence
\begin{equation}
 B_P=-\frac{\sqrt{\Lambda}}{3}D_P.
\end{equation}
Using
\begin{equation}
 \left(\K_\ell^P\right)^{\TF}
 =
 \begin{pmatrix}
 -3M & 2P Q\sqrt{\Lambda}\\
 2P Q\sqrt{\Lambda} & 3M
 \end{pmatrix},
\end{equation}
we obtain
\begin{equation}
 \left(\A_\ell^P\right)^{\TF}
 =
 -\frac{D_P}{6}
 \left(\K_\ell^P\right)^{\TF}.
 \label{eq:BD-key}
\end{equation}
This establishes
Eqs.~\eqref{eq:RN-tree-result}--
\eqref{eq:RN-tree-result-end}.
Thus, for every \(\ell\geq2\) and \(P=\pm1\), we have shown directly
from the Wigner-\(d\) projections that the fixed-source partial-wave
matrix and the RN  Moncrief coupling matrix are diagonalized by the same rotation.

\subsection{First recoil coefficient}

We now derive Proposition~\ref{prop:RN-recoil}. In the
center-of-mass frame, we choose
\begin{align}
 p^\mu&=(E,0,0,-\omega),&
 k^\mu&=\omega(1,0,0,1),\\
 p'^\mu&=(E,-\omega\sin\theta,0,-\omega\cos\theta),&
 k'^\mu&=\omega(1,\sin\theta,0,\cos\theta),
\end{align}
where \(z\equiv\cos\theta\),
\(E=\sqrt{m_\Phi^2+\omega^2}\), and
\(E_{\mathrm{cm}}=E+\omega=\sqrt{s}\).
From these momenta, we obtain the invariants entering
Eq.~\eqref{eq:tree-matrix}:
\begin{equation}
 \begin{gathered}
 D_s=2\omega E_{\rm cm},\qquad
 D_u=-2\omega(E+\omega z),\qquad
 t=-2\omega^2(1-z),\\
 \Xi_\Phi=2\omega^2(1+z)E_{\rm cm}^2,\qquad
 R=\frac{E_{\rm cm}(E+\omega z)}{1-z}.
 \end{gathered}
 \label{eq:cm-recoil-invariants}
\end{equation}
Using \(f_{JI}^{h_fh_i}
 =\mathcal M_{JI}^{h_fh_i}/(8\pi E_{\mathrm{cm}})\) and the parameter relations in Eq.~\eqref{eq:MQ}, we obtain the helicity amplitudes with their recoil dependence kept unexpanded
\begin{align}
 \label{eq:exact-recoil-wave}
 f_{\gamma\gamma}^{++}
 &=M\frac{E_{\rm cm}}{m_\Phi}\frac{\cth^2}{\sth^2}-\frac{Q^2}{M}\frac{m_\Phi}{E+\omega z}\cth^2,
 &f_{\gamma\gamma}^{+-}&=-\frac{Q^2}{M}\frac{m_\Phi^3}{E_{\rm cm}^2(E+\omega z)}\sth^2 ,\\
 f_{gg}^{++}
 &=M\frac{E_{\rm cm}^2}{m_\Phi (E+\omega z)}
   \frac{\cth^4}{\sth^2},
 &f_{gg}^{+-}
 &=M\frac{m_\Phi^3}{E_{\rm cm}^2(E+\omega z)}\sth^2,\\
 f_{\gamma g}^{++}
 &=Q\frac{E_{\rm cm}}{E+\omega z}\frac{\cth^3}{\sth},
 &f_{\gamma g}^{+-}
 &=-Q\frac{m_\Phi^2}{E_{\rm cm}(E+\omega z)}\cth\sth .
 \label{eq:exact-recoil-wave-end}
\end{align}
The Thomson contributions contained in the photon amplitudes above are
\begin{equation}
 f_{\rm Th.}^{++}
 =-\frac{Q^2}{M}\frac{m_\Phi}{E+\omega z}\cth^2,\qquad
 f_{\rm Th.}^{+-}
 =-\frac{Q^2}{M}\frac{m_\Phi^3}{E_{\rm cm}^2(E+\omega z)}\sth^2 .
 \label{eq:exact-recoil-thomson}
\end{equation}
Separating the long-range pole and Thomson contributions and expanding \(E/m_\Phi=1+\Order(\epsilon_{\rm r}^2)\) and
\(E_{\rm cm}/m_\Phi=1+\epsilon_{\rm r}+\Order(\epsilon_{\rm r}^2).\) we obtain Eqs.~
\eqref{eq:rn-wave}--\eqref{eq:rn-wave-end} and \eqref{eq:thomson}--\eqref{eq:thomson-end}.

For each projection integral \(X\), we write
\begin{equation}
 X(\epsilon_{\mathrm r})
 =X^{(0)}
 +\epsilon_{\mathrm r}X^{(\mathrm r)}
 +\Order(\epsilon_{\mathrm r}^2).
\end{equation}

At order \(\epsilon_{\mathrm r}\), the two conversion-channel helicity amplitudes have the same angular factor,
\(2\cth^3\sth\), but their partial-wave projections involve
\(d^\ell_{1,2}\) and \(d^\ell_{1,-2}\), respectively.
Using the finite-sum representation in
Eq.~\eqref{eq:wigner-finite} and the beta-function identity in Eq.~\eqref{eq:beta-reduction}, we obtain
\begin{align}
 \label{eq:recoil-conversion-integrals}
 I_{\mathrm{nonflip}}^{({\rm r})}
 &=\int_{-1}^{1}\dd z\,
   d^\ell_{1,2}\,2\cth^3\sth
   =\frac25\delta_{\ell2},\\
 I_{\rm flip}^{({\rm r})}
 &=\int_{-1}^{1}\dd z\,
   d^\ell_{1,-2}\,2\cth^3\sth
   =-\frac8{(\Lambda+2)\sqrt{\Lambda}}+\frac25\delta_{\ell2}.
 \label{eq:recoil-conversion-integrals-end}
\end{align}
For the diagonal entries the same calculation yields
\begin{align}
 \label{eq:recoil-diagonal-integrals}
 \mathcal J_\gamma^{({\rm r})}
 &=\mathcal J_\gamma^{(0)},\\
 \mathcal J_{g,\mathrm{nonflip}}^{({\rm r})}
 &=\mathcal J_{g,\mathrm{nonflip}}^{(0)}+\frac45\delta_{\ell2},\\
 \mathcal J_{g,\mathrm{flip}}^{({\rm r})}
 &=-\mathcal J_{g,\mathrm{flip}}^{(0)}
   -\frac{24}{(\Lambda+2)\Lambda}+\frac45\delta_{\ell2}.
 \label{eq:recoil-diagonal-integrals-end}
\end{align}
The common forward logarithm in the first two lines again cancels from their
difference.  Therefore
\begin{align}
 \label{eq:recoil-BD}
 B_P^{({\rm r})}
 &=-\frac{8P}{(\Lambda+2)\sqrt{\Lambda}}
   +\frac25(1+P)\delta_{\ell2},\\
 D_P^{({\rm r})}
 &=D_{-P}^{(0)}+\frac{24P}{(\Lambda+2)\Lambda}
   -\frac45(1+P)\delta_{\ell2}.
 \label{eq:recoil-BD-end}
\end{align}
For any real symmetric two-channel matrix in the parameterization
\eqref{eq:BD-matrix},
\begin{equation}
 \comm{\A_\ell^P}{\K_\ell^P}
 =2P M Q\left(\sqrt{\Lambda}D_P+3B_P\right)\Jmat .
 \label{eq:BD-commutator-general}
\end{equation}
The zeroth-order bracket vanishes by \eqref{eq:BD-key}; inserting
\eqref{eq:recoil-BD}--\eqref{eq:recoil-BD-end} into the first-order bracket gives
\eqref{eq:RN-recoil-rho}.

For completeness, the recoil pieces of \eqref{eq:thomson}--\eqref{eq:thomson-end} require only
\begin{equation}
 \int_{-1}^{1}\dd z\,d^\ell_{11}\,z\cth^2
 =\frac15\delta_{\ell2},\qquad
 \int_{-1}^{1}\dd z\,d^\ell_{1,-1}\,
 \sth^2(2+z)=\frac15\delta_{\ell2}.
 \label{eq:recoil-thomson-integrals}
\end{equation}
Because \(\eta_\gamma=-1\), their parity combination adds
\begin{equation}
 \Delta D_P^{({\rm r}),{\rm Th.}}
 =\frac{Q^2}{M^2}\frac{1-P}{5}\delta_{\ell2}.
\end{equation}
Equation~\eqref{eq:BD-commutator-general} then gives
\eqref{eq:RN-recoil-full}.

\subsection{Radial first-Born matching}
Finally, consider the radial equation
\begin{equation}
\left[
 \left( -\frac{\dd^2}{\dd r_*^2}
 +\frac{\Lambda+2}{r^2}\right)\mathbf 1_2
 +\frac{\mathbf W}{r^3}
\right]\boldsymbol{\psi}
=\omega^2\boldsymbol{\psi} \,.
\end{equation}
At this first order in the weak-field tail, the free radial operator uses
\(r_*=r+\Order(M\log(r/M))\), so replacing \(r_*\) by \(r\) inside the Born
integral changes only higher orders in the background strength.  This is why
the integration measure below is \(\dd r\); it is unrelated to a
low-frequency Taylor expansion.
Using the free Riccati--Bessel wave
\(\hat j_\ell(\omega r)=\omega rj_\ell(\omega r)\), its first Born phase is
\begin{align}
 \delta_\ell^{(1)}
 &=-\frac1\omega\int_0^\infty\dd r\,
 \hat j_\ell(\omega r)\frac{\mathbf W}{r^3}\hat j_\ell(\omega r)\notag\\
 &=-\omega \mathbf W\int_0^\infty\frac{\dd x}{x^3}\hat j_\ell(x)^2\notag\\
 &=-\omega \mathbf W\int_0^\infty\frac{\dd x}{x}j_\ell(x)^2
 =-\frac{\omega \mathbf W}{2(\Lambda+2)},
 \label{eq:born-steps}
\end{align}
where \(x=\omega r\) and
\begin{equation}
 \int_0^\infty\frac{\dd x}{x}j_\ell(x)^2
 =\frac{1}{2\ell(\ell+1)}=\frac1{2(\Lambda+2)}.
 \label{eq:bessel}
\end{equation}
The identity holds for \(\ell\geq1\).  Indeed,
\(j_\ell(x)^2/x=\Order(x^{2\ell-1})\) at the origin and
\(\Order(x^{-3})\) at infinity, so the integral is convergent throughout
the radiative range \(\ell\geq2\) used here.
Since \(S_\ell=\mathbf{1}_2+2\ii\delta_\ell+\cdots
=\mathbf{1}_2+\ii\omega\A_\ell+\cdots\), one has
\begin{equation}
 \A_\ell^{(1)}=-\frac{\mathbf W}{\Lambda+2}.
 \label{eq:born-dictionary-detail}
\end{equation}
For RN,
\((\mathbf W_{\rm o})^{\TF}=-(\K_\ell^-)^{\mathrm{TF}}\) and
\((\mathbf W_{\rm e})^{\TF}=-(1+4/\Lambda)(\K_\ell^+)^{\mathrm{TF}}\). Hence, the first born contribution is
\begin{equation}
 \left(\A_{\ell,{0}}^{-}\right)^{\TF}
 =\frac1{\Lambda+2}(\K_\ell^-)^{\mathrm{TF}},\qquad
 \left(\A_{\ell,{0}}^{+}\right)^{\TF}
 =\frac{\Lambda+4}{(\Lambda+2)\Lambda}(\K_\ell^+)^{\mathrm{TF}},
 \label{eq:born-RN-detail}
\end{equation}
which reproduces both lines of \eqref{eq:RN-tree-result}--\eqref{eq:RN-tree-result-end}, including their relative normalization.

\section{Fermion calculation in Dirac form}
\label{app:fermion}

\subsection{Low-energy Dirac reduction}
In this appendix, \(m\equiv m_\Psi\), \(q\equiv q_\Psi\), and \(\epsilon_{\rm r}\equiv\omega/m_\Psi\) with \(\hbar=1\).
Let the incoming source be at rest and choose
\begin{align}
 k^\mu&=\omega(1,0,0,1),\notag\\
 k'^\mu&=\omega'(1,\sin\theta,0,\cos\theta),\notag\\
 \omega'&=\frac{\omega}{1+(\omega/m)(1-\cos\theta)} .
 \label{eq:compton-kinematics}
\end{align}
We keep exact recoil until after separating the spin-even and spin-odd
pieces.  For spin eigenstates along \(\hat z\), define
\begin{equation}
 \M_{\rm av}=\frac{\M_\uparrow+\M_\downarrow}{2},
 \qquad
 \M_{\rm spin}=\frac{\M_\uparrow-\M_\downarrow}{2}.
 \label{eq:spin-average}
\end{equation}
More precisely, at fixed photon helicities \(h_i,h_f\),
\begin{equation}
 \M_\uparrow^{h_fh_i}\equiv
 \M_{s_f=+\hat z,s_i=+\hat z}^{h_fh_i},
 \qquad
 \M_\downarrow^{h_fh_i}\equiv
 \M_{s_f=-\hat z,s_i=-\hat z}^{h_fh_i}.
\end{equation}
Thus, if the amplitude in the source-spin space is
\(\M=A_0\mathbf{1}_2+\bm A\cdot\bm\sigma\), Eq.~\eqref{eq:spin-average}
selects \(A_0\) and \(A_z\).  It is an amplitude-level projection under
reversal of the source polarization, not an unpolarized average of
\(\lvert\M\rvert^2\), and it does not test the off-diagonal source-spin
amplitudes governed by \(A_x,A_y\).

Using the Gordon identity,
\begin{equation}
 \bar u(p')\gamma^\mu u(p)
 =\bar u(p')\left[
 \frac{(p'+p)^\mu}{2m}
 +\frac{\ii\sigma^{\mu\nu}(p'-p)_\nu}{2m}\right]u(p),
 \label{eq:Gordon}
\end{equation}
the convective term is spin even while the Pauli term is odd under
\((\uparrow)~\leftrightarrow~(\downarrow)\).  The reduction can be made explicit as
follows.  Put
\(\hat{\bm n}'=(\sin\theta,0,\cos\theta)\), and choose
\begin{equation}
 \bm\e_i=\frac{(1,\ii h_i,0)}{\sqrt2},
 \qquad
 \bm\e_f^*=\frac{(\cos\theta,-\ii h_f,-\sin\theta)}{\sqrt2}.
\end{equation}
Boldface quantities in this subsection are spatial three-vectors, and their
dot denotes the Euclidean contraction
\(\bm v\cdot\bm w=\sum_{i=1}^3v^iw^i\).
All unbolded four-vector contractions continue to use the mostly-minus
Minkowski metric.  The apparent sign difference between the two conventions
is therefore intentional.
Let \(n^\mu=(1,\hat{\bm z})\), \(n'^\mu=(1,\hat{\bm n}')\), and denote
the electric charge by the same \(q\) as in the main text.  Since
\(p\cdot\e_i=p\cdot\e_f^*=0\), \(\slashed p\,u(p)=m u(p)\), and
\(p\cdot k=m\omega\), \(p\cdot k'=m\omega'\), the two pole diagrams
combine into
\begin{equation}
 \M=\frac{q^2}{2m}\,
 \bar u(p')\left[
 \slashed{\e}_f^*\slashed n\slashed{\e}_i+
 \slashed{\e}_i\slashed n'\slashed{\e}_f^*
 \right]u(p).
\end{equation}
Writing \(\bm\Delta=\bm k-\bm k'
=\omega[\hat{\bm z}-\hat{\bm n}'
+\Order (\epsilon_{\rm r})]\), the canonical
spinors are
\begin{align}
 u(p)&=\sqrt{2m}\begin{pmatrix}\chi\\0\end{pmatrix},\\
 \bar u(p')&=\sqrt{2m}\left[
 \left(
 \chi^\dagger,-\chi^\dagger
 \frac{\bm\sigma\cdot\bm\Delta}{2m}\right)
 +\Order (\epsilon_{\rm r}^2)\right].
\end{align}
For \(\bm e_{\rm out}=\bm\e_f^*\) and
\(\bm e_{\rm in}=\bm\e_i\), Pauli reduction then gives
\begin{align}
 \M_{s_fs_i}^{h_fh_i}
 &=q^2\chi_{s_f}^\dagger
 \left[2\bm e_{\rm out}\cdot\bm e_{\rm in}\,\mathbf{1}_2
 +\frac{\epsilon_{\rm r}}{2}\mathcal K_{\rm Pauli}
 +\Order (\epsilon_{\rm r}^2)\right]\chi_{s_i},\\
 \mathcal K_{\rm Pauli}
 &=\bm\sigma\cdot(\hat{\bm z}-\hat{\bm n}')\,\mathcal Q_{\rm Pauli},\\
 \mathcal Q_{\rm Pauli}
 &=(\bm\sigma\cdot\bm e_{\rm out})\sigma_z
   (\bm\sigma\cdot\bm e_{\rm in})\notag\\
 &\quad+(\bm\sigma\cdot\bm e_{\rm in})
  (\bm\sigma\cdot\hat{\bm n}')
  (\bm\sigma\cdot\bm e_{\rm out}).
\end{align}
The two entries needed by Eq.~\eqref{eq:spin-average} follow from
\begin{equation}
 \bm e_{\rm out}\cdot\bm e_{\rm in}
 =\frac{\cos\theta+h_i h_f}{2},
 \qquad
 \frac12\Tr(\sigma_z\mathcal K_{\rm Pauli})
 =-2h_f(1-\cos\theta)\,
 \bm e_{\rm out}\cdot\bm e_{\rm in}.
\end{equation}
Consequently, with
\(\M_0^{h_fh_i}
=2q^2\bm e_{\rm out}\cdot\bm e_{\rm in}\),
\begin{equation}
 \M_{\uparrow,\downarrow}^{h_fh_i}
 =\M_0^{h_fh_i}\left[
 1\mp h_f\frac{\epsilon_{\rm r}}{2}(1-\cos\theta)
 +\Order (\epsilon_{\rm r}^2)\right],
\end{equation}
where the upper sign is for \(\uparrow\).  Equation~\eqref{eq:spin-average}
therefore yields
\begin{equation}
 \frac{\M_{\rm spin}^{h_fh_i}}
 {\M_{\rm av}^{h_fh_i}}
 =-\sigma_f\frac{\epsilon_{\rm r}}{2}(1-\cos\theta)
 +\Order (\epsilon_{\rm r}^2).
 \label{eq:fermion-ratio}
\end{equation}
Equation~\eqref{eq:fermion-ratio} is not a classical limit of a
spin-\(\tfrac12\) state.  In the units used here, \(a_z=1/(2m)\), and hence
\begin{equation}
 \left.\frac{1}{\omega a_z}
 \frac{\M_{\rm spin}^{h_fh_i}}{\M_{\rm av}^{h_fh_i}}
 \right|_{j=1/2}
 =-\sigma_f(1-\cos\theta)
 \left[1+\Order (\epsilon_{\rm r})\right].
\end{equation}
The spin-odd projection removes the spin-even recoil sector, and division by
\(\omega a_z\) isolates \(c_{10}\) in \eqref{eq:double-expansion}.  By
Proposition~\ref{prop:spin-promotion}, it gives
\eqref{eq:aligned-spin-dressing} after promotion to a classical \(a\).

\subsection{Analytic contraction of the graviton-pole}

The pole contraction entering the elastic photon channel is, up to a common
coupling and propagator normalization,
\begin{equation}
 \M_{\gamma\gamma}^{(g\mathrm{-pole})}
 \propto\frac1t\,
 T_{1/2}^{\mu\nu}
 P_{\mu\nu,\rho\sigma}T_\gamma^{\rho\sigma},
 \label{eq:fermion-hpole-check}
\end{equation}
with
\begin{equation}
 T_{1/2}^{\mu\nu}
 =\frac14\bar u(p')\!
 \left[\gamma^\mu(p+p')^\nu+\gamma^\nu(p+p')^\mu\right]u(p).
 \label{eq:fermion-stress}
\end{equation}
Write \(P=p+p'\), \(\Delta=p'-p=k-k'\), and
\(\sigma^{\mu\nu}=i[\gamma^\mu,\gamma^\nu]/2\).  The Gordon identity gives
\begin{equation}
 T_{1/2}^{\mu\nu}
 =\frac{1}{4m}\bar u(p')
 \left[
 P^\mu P^\nu+\frac{i}{2}
 \left(P^\mu\sigma^{\nu\rho}+P^\nu\sigma^{\mu\rho}\right)\Delta_\rho
 \right]u(p).
 \label{eq:fermion-stress-gordon}
\end{equation}
The first term is spin even and the second is spin odd.  The Maxwell matrix
element, constructed from the incoming and conjugate outgoing field
strengths, obeys
\(\Delta_\rho T_\gamma^{\rho\sigma}=0\) and
\(\eta_{\rho\sigma}T_\gamma^{\rho\sigma}=0\).  Consequently the trace part
of the de Donder numerator drops out, and no gauge-dependent longitudinal
term contributes.

Use
\begin{equation}
 \Delta^\mu=\omega\left[
 (0,-\sin\theta,0,1-\cos\theta)
 +\Order (\epsilon_{\rm r})\right]
\end{equation}
together with the helicity vectors given above.  Denote by
\(\mathcal H_{\uparrow,\downarrow}^{h_fh_i}\) the numerator in
\eqref{eq:fermion-hpole-check} for source spin along \(\pm\hat z\), and define
\(\mathcal H_{\rm av}=(\mathcal H_\uparrow+\mathcal H_\downarrow)/2\) and
\(\mathcal H_{\rm spin}=(\mathcal H_\uparrow-\mathcal H_\downarrow)/2\).
Direct contraction of \eqref{eq:fermion-stress-gordon} gives
\begin{align}
 \mathcal H_{\rm av}^{h_fh_i}
 &=\mathcal N\left[
   2(1+\cos\theta)\,\delta_{h_fh_i}
   +\Order (\epsilon_{\rm r})\right],\notag\\
 \mathcal H_{\rm spin}^{h_fh_i}
 &=\mathcal N\left[
   -\sigma_f\epsilon_{\rm r}
   (1-\cos^2\theta)\,\delta_{h_fh_i}
   +\Order (\epsilon_{\rm r}^2)\right],
 \label{eq:hpole-even-odd}
\end{align}
where \(\mathcal N\) is the common spin-independent normalization, including
the powers of \(\omega\) suppressed in
\eqref{eq:fermion-hpole-check}.  Hence
\begin{equation}
 \left.
 \frac{\mathcal H_{\rm spin}^{h_fh_i}}
      {\mathcal H_{\rm av}^{h_fh_i}}
 \right|_{h_f=h_i}
 =-\sigma_f a_z\omega(1-\cos\theta),
 \qquad
 \mathcal H_{\rm spin}^{h,-h}=0 .
 \label{eq:hpole-analytic-ratio}
\end{equation}
with \(a_z=1/(2m)\) and \(\sigma_f\equiv\operatorname{sgn}(h_f)\).  This proves analytically that the graviton-pole
has the same aligned factor \(\mathcal D_{fi}\) as the QED kernel.

\subsection{Full-tree factorization}

It remains to show that the same factor appears in the conversion and
gravitational Compton entries, including their local diagrams.  Let
\(\mathcal C_j\) denote the complete gauge-theory Compton kernel for a
massive source of spin \(j\), and let $N_i,N_f$, and $R$ be the kinematic factors
defined in Appendix~\ref{app:spinor-scalar}.  The factor customarily denoted
by \(F\) in the gravitational-Compton factorization literature is, in our
invariants,
\begin{equation}
 F\equiv
 \frac{(p\cdot k)(p\cdot k')}{k\cdot k'}
 =\frac{D_sD_u}{2t}=R .
 \label{eq:F-equals-R}
\end{equation}
Writing
\(\M_{\gamma\gamma,{\rm QED}}^{[j]}=q^2\mathcal C_j\), summing the massive
\(s\)- and
\(u\)-channel graphs, the mixed photon-graviton seagull, and the graviton
pole, and then using the on-shell Dirac equations and the two Ward identities,
gives
\begin{align}
 \M_{\gamma\rightarrow g}^{[1/2]}
 &=\frac{\kappa q}{2}N_f\mathcal C_{1/2},&
 \M_{g\rightarrow\gamma}^{[1/2]}
 &=\frac{\kappa q}{2}N_i\mathcal C_{1/2},\notag\\
 \M_{gg}^{[1/2]}
 &=\frac{\kappa^2}{8}R\,
   \mathcal C_{1/2}\,\mathcal C_0 .
 \label{eq:fermion-full-factorization}
\end{align}
Equation~\eqref{eq:fermion-full-factorization} is an identity for the
\emph{full} minimal tree-level amplitudes, not merely for their pole residues
\cite{Choi:1994ax,Holstein:2006bh}.  Algebraically, the denominators are
first brought to the common denominator \(D_sD_u\); the Dirac equation
removes \(\slashed p-m\) and \(\slashed p'-m\), while
\(k\cdot\varepsilon_i=k'\cdot\varepsilon_f^*=0\) groups the remaining numerator
into \(\mathcal C_{1/2}\).  The seagull and graviton-pole numerators supply
precisely the terms needed to replace a second spin-dependent copy by the
scalar kernel \(\mathcal C_0\).  This is why the last line contains
\(\mathcal C_{1/2}\,\mathcal C_0\), rather than
\(\mathcal C_{1/2}^2\).

The factorization can also be checked channel by channel.  Since
\(R=D_sD_u/(2t)\), its massive-channel residues are
\begin{equation}
 \underset{D_s=0}{\operatorname{Res}}\,
 \M_{gg}^{[1/2]}
 =\frac{\kappa^2D_u}{16t}
 \left(\underset{D_s=0}{\operatorname{Res}}\mathcal C_{1/2}\right)
 \left(\underset{D_s=0}{\operatorname{Res}}\mathcal C_0\right),
 \qquad (s\leftrightarrow u),
\end{equation}
which are the products of the corresponding gravitational three-point
amplitudes.  The graviton-channel residue and the local completion are fixed
by the same Ward identities; a direct diagram sum yields
\eqref{eq:fermion-full-factorization}.  This step assumes minimal coupling:
an anomalous Pauli interaction beyond the Dirac \(g=2\) moment,
curvature--field-strength operators, and tidal operators would add
independent local terms.

Combining \eqref{eq:aligned-spin-dressing},
\eqref{eq:hpole-analytic-ratio}, and
\eqref{eq:fermion-full-factorization} proves
Proposition~\ref{prop:linear-spin-factorization}.

\section{Aligned \texorpdfstring{\(J=2\)}{J=2} data and their limitations}
\label{app:aligned}

Put \( z = \cos\theta\). For \(J=2\), the needed functions include 
\begin{align}
 d^2_{11}&=\frac{(1+z)(2z-1)}2,&
 d^2_{1,-1}&=\frac{(1-z)(2z+1)}2,\\
 d^2_{12}&=2\cth^3\sth,&
 d^2_{1,-2}&=-2\cth\sth^3,\\
 d^2_{22}&=\cth^4,&
 d^2_{2,-2}&=\sth^4 .
 \label{eq:d2}
\end{align}
Direct integration of \eqref{eq:delta-f}--\eqref{eq:delta-f-end} gives
\begin{align}
 \left.\int\dd z\,d^2_{11}\delta f_{\gamma\gamma}^{++}
 \right|_{\rm Th.}
 &=-\frac{a\omega Q^2}{5M},\notag\\
 \left.\int\dd z\,d^2_{1,-1}\delta f_{\gamma\gamma}^{+-}
 \right|_{\rm Th.}
 &=-\frac{a\omega Q^2}{5M},\\
 \int\dd z\,d^2_{12}\delta f_{\gamma g}^{++}
 &=-\frac{2a\omega Q}{5},\notag\\
 \int\dd z\,d^2_{1,-2}\delta f_{\gamma g}^{+-}
 &=-\frac{2a\omega Q}{5},\\
 \int\dd z\,d^2_{22}\delta f_{gg}^{++}
 &=-\frac{4a\omega M}{5},\notag\\
 \int\dd z\,d^2_{2,-2}\delta f_{gg}^{+-}
 &=-\frac{4a\omega M}{5}.
 \label{eq:J2-integrals}
\end{align}
The graviton-pole contribution to the elastic-photon diagonal entry has a vanishing \(J=2\) projection, whereas the spin-dressed Thomson contribution has support only for \(J\leq2\).

For later use, define the \(J=2\) Wigner-\(d\) integrals by
\begin{equation}
I_{XY}^{h_fh_i}\equiv
\int_{-1}^{1}\dd z,d^2_{s_Xh_f,s_Yh_i}(\theta)\delta f_{XY}^{h_fh_i}(z),
\end{equation}
where \(X,Y\in\{\gamma,g\}\) label the final and initial species, respectively, and \(s_\gamma=1\), \(s_g=2\).
If one mechanically applies the spherical RN parity combination separately
to these aligned helicity data, the photon, graviton, and conversion entries
combine, respectively, as
\begin{equation}
I_{\gamma\gamma}^{++}-P I_{\gamma\gamma}^{+-},\qquad I_{gg}^{++}+P I_{gg}^{+-},\qquad P I_{\gamma g}^{++}+I_{\gamma g}^{+-}.
\end{equation}
where the last expression includes the
\(\mathbf U_P=\operatorname{diag}(P,1)\) channel rephasing.  One obtains
\begin{align}
 \label{eq:formal-J2}
 \delta\A_{2,+}^{\rm formal}
 &=\frac{a\omega}{5}
 \begin{pmatrix}0&-4Q\\-4Q&-8M\end{pmatrix},\\
 \delta\A_{2,-}^{\rm formal}
 &=\frac{a\omega}{5}
 \begin{pmatrix}-2Q^2/M&0\\0&0\end{pmatrix}.
 \label{eq:formal-J2-end}
\end{align}
Equations~\eqref{eq:formal-J2}--\eqref{eq:formal-J2-end} also isolate the origin of the two parity
blocks.  
For \(P=-1\) (axial parity), the equal \(++\) and \(+-\) graviton integrals cancel,
as do the two conversion integrals, while the long-range photon term has zero
\(J=2\) projection.  The remaining entry
\(-2a\omega Q^2/(5M)\) is entirely the Thomson term dressed by
\(1-\cos\theta\). For \(P=+1\) (polar parity), the two Thomson integrals cancel instead.
Thus the axial result is generated by the degree-raising leakage
\begin{equation}
 J=1\ \hbox{Thomson polynomial}
 \ \xrightarrow{\ \times(1-\cos\theta)\ }\ J\leq2,
\end{equation}
not by the long-range sector that entered the spinless
\(\ell\geq2\) theorem.
Their trace-free parts do not commute with the RN matrices:
\begin{align}
 \label{eq:formal-comms}
 \comm{\K_2^+}{(\delta\A_{2,+}^{\rm formal})^{\TF}}
 &=-\frac{8a\omega MQ}{5}\Jmat,\\
 \comm{\K_2^-}{(\delta\A_{2,-}^{\rm formal})^{\TF}}
 &=-\frac{8a\omega Q^3}{5M}\Jmat .
 \label{eq:formal-comms-end}
\end{align}
The nonzero commutators show that the unchanged RN projectors fail for this
formal \(J=2\) block; Section~\ref{sec:spin} states why this restricted result
is not a KN separability test.

\bibliography{references}
\bibliographystyle{utphys28mod}

\end{document}